\documentclass[letterpaper,onecolumn,11pt,accepted=2026-10-11]{quantumarticle}
\pdfoutput=1
\usepackage[utf8]{inputenc}
\usepackage[english]{babel}
\usepackage[T1]{fontenc}

\usepackage{amsmath,amssymb,amsfonts,amsthm,mathtools,mathrsfs,bm,bbm,dsfont}
\usepackage{graphicx,url,tcolorbox,soul}
\usepackage{xcolor}
\usepackage{tikz}
\usepackage{multirow}
\usepackage{booktabs}
\usepackage{cases}
\usepackage{enumitem}
\usepackage{fancybox}
\usepackage{algorithm}
\usepackage{verbatim}
\usepackage{algorithmic}
\usepackage{braket}
\usepackage[normalem]{ulem}
\usepackage{scalerel,stackengine}
\usepackage{xspace}
\usepackage{array}
\usepackage{tabularx}
\usepackage{adjustbox}
\usepackage{hyperref}
\hypersetup{
    colorlinks,
    linkcolor={blue!80!black},
    citecolor={red!80!black},
    urlcolor={blue!80!black}
}

\newcommand\redout{\bgroup\markoverwith{\textcolor{red}{\rule[0.5ex]{2pt}{0.8pt}}}\ULon}

\stackMath
\newcommand\reallywidehat[1]{
\sav{\tmpbox}{\stretchto{%
0  \scaleto{%
    \scalerel*[\widthof{\ensuremath{#1}}]{\kern-.6pt\bigwedge\kern-.6pt}%
    {\rule[-\textheight/2]{1ex}{\textheight}}
  }{\textheight}%
}{0.5ex}}%
\stackon[1pt]{#1}{\tmpbox}%
}

\newcommand{\floor}[1]{\left \lfloor #1 \right \rfloor}

\newtheorem{theorem}{Theorem}
\newtheorem{cor}[theorem]{Corollary}

\newtheorem{lemma}[theorem]{Lemma}
\newtheorem{prop}[theorem]{Proposition}

\theoremstyle{remark}
\newtheorem{definition}{Definition}
\newtheorem{remark}{Remark}
\newtheorem{construction}{Construction}

\newenvironment{fminipage}%
  {\begin{Sbox}\begin{minipage}}%
  {\end{minipage}\end{Sbox}\fbox{\TheSbox}}

\makeatletter
\newcommand*{\rom}[1]{\expandafter\@slowromancap\romannumeral #1@}
\makeatother

\DeclareMathOperator{\supp}{supp}
\DeclareMathOperator{\Span}{Span}
\newcommand{\var}{\mathrm{Var}}
\newcommand{\cov}{\mathrm{cov}}

\newcommand{\nc}{\newcommand}
\nc{\un}{{\underline{n}}}  \nc{\ux}{{\underline{x}}}  \nc{\uX}{{\underline{X}}}  \nc{\uY}{{\underline{Y}}}  
\nc{\uy}{{\underline{y}} } \nc{\ue}{{\underline{e}}}  \nc{\uf}{{\underline{f}} } \nc{\ur}{{\underline{r}}} \nc{\Corr}{\mathrm{corr}}
\nc{\1}{{\mathbbm{1}}}
\nc\bfx{\boldsymbol x}
\nc{\Cov}{\mathrm{Cov}}

\newcommand{\Ind}{\mathbbm{1}}

\newcommand{\tr}{\mathrm{tr}}

\newcommand{\abs}[1]{\left|#1\right|}
\newcommand{\R}{\mathbb{R}} 
\newcommand{\N}{\mathbb{N}}

\newcommand{\Z}{\mathbb{Z}}

\newcommand{\C}{\mathbb{C}}

\newcommand{\E}{\mathbb{E}}

\newcommand{\cC}{\mathcal{C}} \nc\sC{{\mathscr C}}
\newcommand{\cE}{\mathcal{E}}  
\def\P{{\mathbb P}}

\newcommand{\calA}{{\cal A}}
\newcommand{\calB}{{\cal B}}

\newcommand{\calD}{{\cal D}}
\newcommand{\calE}{{\cal E}}
\newcommand{\calF}{{\cal F}}

\newcommand{\calH}{{\cal H}}

\newcommand{\calJ}{{\cal J}}

\newcommand{\calM}{{\cal M}}
\newcommand{\calN}{{\cal N}}

\newcommand{\calP}{{\cal P}}
\newcommand{\calQ}{{\cal Q}}

\newcommand{\calS}{{\cal S}}

\DeclarePairedDelimiterX{\cond}[1]{[}{]}{\setargs{#1}}
\NewDocumentCommand{\setargs}{>{\SplitArgument{1}{;}}m}
{\setargsaux#1}
\NewDocumentCommand{\setargsaux}{mm}
{\IfNoValueTF{#2}{#1} {#1\,\delimsize|\,\mathopen{}#2}}

\newcommand{\be}{\begin{equation}}
\newcommand{\ee}{\end{equation}}
\newcommand{\beqna}{\begin{eqnarray}}
\newcommand{\eeqna}{\end{eqnarray}}

\newcommand{\p}[1]{\left(#1\right)}
\newcommand{\pp}[1]{\left[#1\right]}
\newcommand{\ppp}[1]{\left\{#1\right\}}
\newcommand{\norm}[1]{\left\|#1\right\|}

\newcommand{\s}[1]{\mathsf{#1}}
\newcommand{\sr}[1]{\mathrm{#1}}
\newcommand{\su}[1]{\underline{\mathsf{#1}}}

\allowdisplaybreaks

\begin{document}

\title{Asymptotically good bosonic Fock state codes}

\author{Dor Elimelech}
\orcid{0000-0002-9141-8669}
\email{dor@umd.edu}
\affiliation{Joint Center for Quantum Information and Computer Science (QuICS), UMD/NIST}
\affiliation{Institute for Systems Research, University of Maryland, College Park, MD 20742}
\author{Arda Aydin}
\orcid{0009-0003-9960-1841}
\email{arda@umd.edu}
\affiliation{Institute for Systems Research, University of Maryland, College Park, MD 20742}
\author{Alexander Barg}
\orcid{0000-0002-8972-4413}
\email{abarg@umd.edu}
\affiliation{Institute for Systems Research, University of Maryland, College Park, MD 20742}
\affiliation{Joint Center for Quantum Information and Computer Science (QuICS), UMD/NIST}
\affiliation{Department of ECE, University of Maryland, College Park, MD 20742}

\maketitle

 \begin{abstract}
We study the error-correction properties of multi-mode Fock-state codes under amplitude-damping (AD) noise, focusing on the asymptotic regime in which the total excitation of the code states grows without limit and the number of photon losses induced by the noise scales linearly with it. In this setting, existing code families, which correct only sublinearly many photon losses, do not protect against amplitude-damping (AD) noise with a constant loss parameter. As our main result, we address this gap by constructing asymptotically good Fock-state codes relying on random classical codes in the discrete simplex.  Unlike many standard constructions of random quantum codes, our construction introduces randomness through the underlying classical indexing structure. Randomization also enables another desirable feature: bounded per-mode occupancy, which limits the number of photons in any individual mode and thereby increases the coherence lifetime of the code states. Finally, via a relation between Fock-state codes and permutation-invariant codes, our results also yield asymptotically good families of qudit permutation-invariant codes as well as codes in monolithic nuclear state spaces.
 \end{abstract}
\newpage
{\small
\tableofcontents}

\section{\label{sec:Introduction} Introduction}

Large-scale quantum computing is essential for running advanced quantum algorithms. Maintaining reliability at scale requires error correction codes that satisfy specific criteria imposed by hardware limitations. While various state spaces can be used for that purpose, bosonic systems provide a promising approach by encoding information in the infinite-dimensional Hilbert space of a harmonic oscillator. In this work, we study Fock state codes in the asymptotic regime when the total excitation scales to infinity.

Fock states, a.k.a. number states, are the eigenstates of the number operator $\hat{n}$, forming an infinite basis for the Hamiltonian of the harmonic oscillator. A Fock state $\ket{n}, n=0,1,2,\ldots$ represents the number of photons occupying a specific mode of the electromagnetic field. Fock state codes are designed to protect quantum information by encoding it in the Fock basis. The encoded information in such systems is subject to decoherence resulting from photon loss. This noise model corresponds to the \textit{amplitude damping (AD)} channel, $\calN$, and admits a Kraus representation
\begin{gather} \label{eq: AD channel}
\rho \mapsto \sum_{k=0}^{\infty}A_k\rho A_k^\dagger\\
\intertext{
with Kraus operators}
 \label{eq: AD channel Kraus}
    A_k=\sum_{n=k}^{\infty} \sqrt{\binom{n}{k}} \sqrt{(1-\gamma)^{n-k} \gamma^k}|n-k\rangle\langle n|, \quad k=0,1,2,\ldots .
\end{gather}
Here, $\gamma\in(0,1]$ is a constant denoting the probability of decay, and the integer $k$ represents the number of photons lost. Fock state codes can be generalized to a multi-mode setting by utilizing $q$  harmonic oscillators rather than a single one. The $q$-mode Fock space is spanned by the product states $\ket{n_0}\otimes \ket{n_1}\otimes \ldots \otimes \ket{n_{q-1}}, n_i=0,1,\ldots; i=0,1,\ldots,q-1$.  In this paper, we focus on \textit{constant excitation} Fock state codes, in which every state satisfies the fixed excitation condition $\sum_{i=0}^{q-1}n_i=N$.

We study error correction with multi-mode Fock state codes used over the AD channel $\calN$. For this noise model, it is typically assumed that the channel acts independently on each mode of the composite system. Error-correcting properties of Fock state codes are often evaluated
relying on a quantity termed code fidelity, which serves as a computationally tractable proxy for the worst-case fidelity, $\calF$. Given a quantum code $Q$, the code fidelity is defined as 
        \[
   \calF_{\calA,Q} := \min_{\ket{\psi}\in Q}\sum_{A\in \calA}\bra{\psi}A^\dagger A\ket{\psi},
  \]
where $\calA$ is the set of errors that are corrected by the code $Q$. 
For the AD channel, we can define $\calA$ as the set of operators $\calA_{\leq t}$ representing at most $t$ photon losses. 
For the case of constant excitation codes, prior works rely on an approximation of the code fidelity  $\calF\geq 1-O(\gamma^{t+1})$, supporting the claim of error correction for fixed-length codes in channels with a small loss parameter.
 This approach, initially proposed in  \cite{chuang1997bosonic} and used in a number of follow-up papers \cite{wasilewski2007protecting,leung1997approximate,OuyangADCode}, has been a common starting point in the code analysis. 
 We argue that the utility of code fidelity is limited for
two distinct reasons. First, in the case of constant excitation codes, considered here and in \cite{chuang1997bosonic,aydin2025quantum},  $\calF_{\calA_{\le t}}$ is a universal constant that does not depend on the code (see Prop.~\ref{prop:ConstExFidel}), and second, code fidelity as a performance
metric is of little use unless the sequence of codes $(Q_N)_N$
recovers from a proportion of photon losses linear in $N$. It is easily shown, as we argue below, that
codes that do not satisfy this condition cannot perform well on the AD channel unless $\gamma$ vanishes as $N$ increases.

This leaves the challenge of constructing positive rate Fock state codes that correct a linear proportion of losses, typically referred to as \textit{asymptotically good}, which until this work was left unanswered. 
 As far as code constructions are concerned, for a long while
the best known scaling of the number of correctable errors $t$ for codes of total excitation $N$ has been $t\propto \sqrt{N}$, due to a construction of Bergmann and van Loock \cite{VanLoock}. A recent paper \cite{aydin2025quantum} took another step further by introducing a general construction of multi-mode constant excitation Fock state codes relying on existing classical $\ell_1$-simplex codes, used to produce positive rate codes capable of correcting $o(N/\log N)$ errors. While the codes presented in \cite{aydin2025quantum} offer a significant advance in error-correction capability, they suffer from three major  pitfalls:  
\begin{enumerate}
    \item The sublinear photon-loss correction implies no protection against the AD channel when $N$ is large.
    \item Their construction relies on a result from convex geometry, whose implementation incurs computation cost that scales doubly exponentially with the total number of photons, rendering the construction virtually unfeasible.
    \item They lack a desirable feature of Fock states: a limited photon count per mode (balancedness). Deterministic generation of Fock states with higher photon numbers remains a significant experimental bottleneck for realizing scalable Fock state codes \cite{PhysRevLett.125.093603, Waks_2006, Hofheinz2008-lz}. Furthermore, high-photon-number states are also more fragile: the probability of losing a single photon from $\ket{n}$ is approximately $n$ times higher than the $\ket{1}$ state, so the coherence lifetime of a Fock state $\ket{n}$ scales as $1/n$  \cite{Wang2008DecayofFock}. This phenomenon limits the maximum number of photons that can be utilized in a single mode for realizing Fock state codes. 
\end{enumerate}

In this work, we overcome these shortcomings by (1) replacing the classical deterministic $\ell_1$-simplex codes employed in  \cite{aydin2025quantum} with randomized constructions and (2) replacing exact error-correction mechanisms with \textit{approximate} ones. Harnessing the resulting randomness in the obtained Fock-state code, we dispense with the convex geometry argument
in favor of a simple low-complexity averaging procedure, thereby overcoming the complexity obstacle. In addition, these random $\ell_1$ codes have better parameters than the deterministic $\ell_1$ codes used in \cite{aydin2025quantum}, which enables approximate correction of linearly many photon losses and supports strong protection against the AD channel. 

This set of results constitutes the main contribution of this work. We next describe the ideas
that enabled us to establish these results, highlighting the gains attained by relying on approximate quantum error correction.
Although relaxing the exact quantum error correction (QEC) requirement to an approximate one ostensibly reflects a performance loss, this is in fact not the case. Indeed, even exact QEC of $t$-amplitude errors ultimately translates only to approximate QEC for the AD channel that is not restricted to a fixed number of photon losses. To make this perspective rigorous, we formally define a truncated AD channel $\calN_{\le t}$ and relate the worst-case entanglement fidelity of the code used on it to its performance on the non-truncated AD channel. Specifically, we prove that a code is an approximate quantum error-correcting (AQEC) code for the AD channel if and only if it is an (exact or approximate) QEC code for the truncated AD channel $\calN_{\le t}$. Our formalization enables us to extend the definition of asymptotically good families to the approximate correction case (see Def.~\ref{def:  AAG codes}). Coupled with the equivalence of exact and approximate error correction on the AD channel $\calN$ (Theorem~\ref{th:TruncToAD}), this leads to a conclusion that a code sequence is asymptotically good in either approximate or exact sense on the $\calN_{\le t}$ channel if and only if it is approximate asymptotically good for the AD channel (see Corollary~\ref{cor:operative}). 

To prove that our codes are indeed asymptotically good, we use the framework of approximate error correction developed by B\'eny and Oreshkov \cite{beny2010general}, who derived necessary and sufficient conditions for a code to be $\epsilon$-error-correcting for a noisy channel. The averaging procedure over the random classical $\ell_1$ codewords together with concentration arguments ensures that with probability exponentially close to $1$, the resulting quantum code satisfies the AQEC conditions for linearly many photon losses while maintaining a positive rate. We also note that, although our randomized construction only succeeds with high probability, once a realization is given, one can efficiently verify whether the resulting code is good by performing a number of checks proportional to $(\text{number of Kraus operators})^2 \times (\text{dimension of the code})^2$. 

More broadly, this fits the general theme that random quantum code constructions are natural candidates for approximate quantum error correction. However, unlike existing works in which randomness is introduced at the level of the physical Hilbert space, such as Haar-random codes or constructions based on randomly selected many-body eigenstates \cite{klesseApproximateQuantumError2007,kong2022near,ma2025haar,brandao2019quantum}, our construction introduces randomness at the level of the underlying classical indexing structure from which the quantum codewords are generated. This is advantageous because the resulting quantum code inherits useful properties of random classical codes. In particular, relying on randomization of classical $\ell_1$ codes, we obtain Fock state codes that satisfy the balancedness property in a strong form, with at most logarithmic per-mode excitation.

The paper is structured as follows. We start with a formalization of the truncated AD channel $\calN_{\le t}$ and quantify its relation to the ``full'' AD channel $\calN$. In this part (Section~\ref{sec: truncated}), we also establish a relation between exact and approximate error correction and define asymptotically good Fock state code families. Section~\ref{sec:FockCodesConst} is devoted to the construction of 
the described Fock state code families. Supporting results about classical $\ell_1$ codes are 
collected in Section~\ref{sec:ClasicalELl_1}. Finally, in Section~\ref{sec: PI Spin} we recall that one of the main results of \cite{aydin2025quantum} establishes a relation
 between Fock state codes and qudit permutation invariant (PI) codes \cite{ruskaiExchange,ouyangPI,aydin2023family} as well as families of spin codes \cite{gross}. Leveraging our results for Fock state codes, we prove the existence of 
asymptotically good codes from these code families, advancing the best previously known scaling of $O(N/\log N)$ for the
number of correctable errors \cite{aydin2025quantum}, where $N$ is the code length for PI codes and total spin for spin codes, respectively.

\section{Preliminaries}
\subsection{Notation and definitions}
We begin by introducing the notation used throughout the paper. Underlined letters $\un=(n_0,\dots,n_{q-1})$ refer to $q$-tuples of nonnegative integers. Random variables are denoted by sans-serif symbols (e.g., $\s{X}, \s{Y}$). We use $\Z_0$ and $\N$ to denote the set of non-negative and positive integers, respectively. For $K\in \N$ we define  $[K]\triangleq \ppp{0,1,\dots,K-1}$. For an integer $N,q\in \N  $, consider the set of all ordered integer partitions (a discrete simplex), given by
   \[
   \calS_{q,N}\triangleq \Big\{\un\in \Z_0^{q} ~:~ \sum_{i=0}^{q-1} n_i=N\Big\}.
   \]
The number of points in the simplex is given by the balls-and-bins calculation, 
   \[
   \abs{\calS_{q,N}}=\binom{N+q-1}{q-1}.
   \]
   By abuse of notation, we will use the notation $\calS_{\alpha,N}$ also whenever $\alpha>1$ is not an integer, with the convention that $\calS_{\alpha,N}\triangleq \calS_{\floor{\alpha},N}$.
For $N\in \N$ and $\un\in \Z^q$, the multinomial coefficient $\binom{N}{\un}$ is defined in a standard way:
\[\binom{N}{\un}\triangleq \begin{cases}
\frac{N!}{\prod_{i=0}^{q-1} n_i!} & \un\in \calS_{q,N}\\
0 & \text{otherwise.}
\end{cases}\]
We will often use the binary entropy function $h_2:[0,1]\to[0,1]$, defined as
    \[
    h_2(x)=-x\log_2 x-(1-x)\log_2(1-x),\label{eq:binaryEnt}
    \]
to approximate the growth rate of binomial coefficients.

Given a Hilbert space $\calH$, denote by $L(\calH)$ and $D(\calH)$ the space of bounded linear operators and density operators on it, respectively. We consider the Hilbert space $\ell_2(\N)$ of all sequences $\mathbf{a}=(a_n)_{n=0}^\infty$ with finite $\ell_2$ norm: $\norm{\mathbf{a}}_2=\sqrt{\sum_{n}|a_n|^2}< \infty $. We denote the standard orthonormal basis of $\ell_2(\N)$ by $\ppp{\ket{0},\ket{1},\dots}$, where $\ket{n}$ denotes the sequence $\boldsymbol{a}^{(n)}$ given by $a^{(n)}_i=\delta_{n}(i)$, with $\delta_n$ denoting the Kronecker delta at $n$.

We use standard asymptotic notation: for functions $f,g:\N \to \R_+$, we say that $f=O(g)$ if there exists a constant $c>0$ such that $f(n)\leq c g(n)$ for all $n\in \N$. We say that $f=o(g)$ if $\lim_{n\to\infty}f(n)/g(n)=0$.

\subsection{Exact and approximate quantum codes} 
A quantum channel is given by a completely positive trace-preserving (CPTP) map  $\calN:L(\calH)\to L(\calH')$, where $\calH$ and $\calH'$ are two Hilbert spaces. Such CPTP maps can be equivalently represented by a set of Kraus error operators $\ppp{A_k}_k$ satisfying the completeness condition $\sum_{k}A_k^\dag A_k =I$, such that the action of $\calN$ 
on a state $\rho$ has the form $\calN(\rho)=\sum_{k}A_k \rho A_k^\dag$. A code $Q$ in $\cal H$ is said to correct an error set $\calE$ if there exists a decoding operation (i.e., a CPTP map) $\calD:L(\calH')\to L(\calH) $ such that for any channel $\calN$ with $\{A_k\}_k\subset\calE$ we have 
\[\calD \circ \calN |_{L(C)}=I_{L(C)},\]
meaning that any input state $\ket{\psi_{\mathrm{in}}}\in Q$ can be perfectly recovered from its noisy version. 
The well-known Knill-Laflamme (KL) conditions imply that a code $Q$ with an orthonormal basis $\ket{c_0},\dots,\ket{c_{k-1}}$ corrects $\calE$ if there are constants $\ppp{\lambda_{k,l}}_{k,l}$, such that for any $A_k,A_l\in \cE$ and $i,j$ we have 
\begin{align}
    \bra{c_i}A_k^\dag A_l \ket{c_j}=\lambda_{k,l}\cdot \delta_{i,j}.\label{eq:KLcond}
\end{align}

Under approximate quantum error-correction (AQEC), we allow the recovered state to be slightly different, but still close,  to the original state. Many frameworks for AQEC were studied in the literature \cite{mandayam2012towards,schumacher2002approximate,beny2009conditions,ma2025haar, beny2010general}, suggesting various criteria and tools to gauge the performance of an AQEC code. We follow the formalism introduced by B\'eny and Oreshkov \cite{beny2010general}, in which the performance of an AQEC code is measured using the worst-case entanglement fidelity, defined as
   \begin{equation}\label{eq: EF}
   \calF_e(\calN,\calM)\triangleq \inf_{\rho\in D(\calH)}\calF\p{\calN\otimes I_{L(\calH)}(\ket{\psi_\rho}\bra{\psi_\rho}),\calM\otimes I_{L(\calH)}(\ket{\psi_\rho}\bra{\psi_\rho})},
   \end{equation}
    where $ \calN, \calM :L(\calH)\to L(\calH')$ are two quantum channels, $\ket{\psi_{\rho}}$ is any purification of $\rho$, and $\calF(\rho,\tau)=\norm{\sqrt{\rho}\sqrt{\tau}}_1$ is the fidelity of a pair of states (the trace norm $\norm{\cdot}_1$ is introduced below in Definition~\ref{def:norms}). While entanglement fidelity provides a natural measure of approximation between channels, it is often technically convenient to work with an equivalent figure of merit. Below, we quantify approximation using the Bures distance, defined as
\[ d(\calN,\calM):=\sqrt{1-\calF_e(\calN,\calM)}.\]
\begin{tcolorbox}[width=1\linewidth, left=1mm, right=1mm, boxsep=0pt, boxrule=0.5pt, sharp corners=all, colback=white!95!black] 
\begin{definition}[Approximate quantum error correction]\label{def:AQEC}
     A code $Q\subseteq \calH$ is called an $\varepsilon$-approximate error-correcting code (AQECC) for a channel $\calN$ if there exists a decoding operation $\calD$ such that
    \[d\p{\calD \circ \calN|_{L(Q)},I_{L(Q)}}=\sqrt{1-\calF_{e}\p{\calD \circ \calN|_{L(Q)},I_{L(Q)}}}\leq  \varepsilon.\]
Extending this definition, consider a sequence of channels $(\calN)_N$ on Hilbert spaces $(\calH_N)_N$ such that there exists
a sequence of codes $Q_N\subset \calH_N$ with the property that for each $N,$ $Q_N$ is an $\epsilon_N$-AQECC for
$\calN_N$. If $\varepsilon_N\to 0$ as $N\to\infty$, we say that $\calQ=(Q_N)_N$ is an approximate error-correcting sequence  of codes for $(\calN_N)_N$.
\end{definition}
\end{tcolorbox}
We observe that a code that is an $\varepsilon$-AQECC for a channel $\calN$ becomes an exact code if $\varepsilon=0$.
\begin{remark}
    In contrast to exact quantum error correction, the approximate case does not satisfy the linearity property with respect to the set of correctable errors. For exact quantum codes, the KL conditions imply that if a code corrects a set of error operators, then it also corrects their linear span, making the error-set model rich and meaningful. In the approximate setting, this equivalence fails: approximate correctability of a given set of error operators does not extend in a controlled way to linear combinations or rescalings, and it is therefore unclear which quantum channels can be meaningfully composed from a fixed approximately correctable error set. As a consequence, in the approximate case, we switch perspective from 
    the error set to the channel view, specifying a noise model, and determining whether it can be approximately inverted.
   This approach generally results in a systematic formulation of AQEC. One exception is presented in a recent work \cite{ma2025haar}, which develops an approximate error-set model; however, it relies crucially on the strong properties of Haar random codes and highly regular unitary error families and does not readily extend beyond this special setting.
\end{remark}

In \cite{beny2010general}, B\'eny and Oreshkov derived necessary and sufficient conditions for approximate quantum error correction on a given channel, formulated in terms of its Kraus operators:

\begin{theorem}[AQEC conditions \cite{beny2010general}]\label{th:BenyOreshkov}
    Let $\calN$ be a channel with Kraus operators $\calE=\ppp{A_k}_{k=0}^{M-1}$. A code space $Q\subset \calH$ is an $\varepsilon$-AQECC if and only if there {exists a density matrix  $\lambda=\ppp{\lambda_{k,l}}_{k,l}$} such that $d(\Lambda, \Lambda+\calB)\leq \varepsilon$, where  $\Lambda,\calB:L(\calH)\to L(\C^{M})$ are the quantum channels given by 
    \begin{equation}\label{eq: Lambda B}
    \Lambda(\rho)=\sum_{k,l=0}^{M-1} \lambda_{k,l}\tr(\rho)\ket{k}\bra{l}, \quad \calB(\rho)=\sum_{k,l=0}^{M-1} \tr(B_{k,l}\rho)\ket{k}\bra{l},
    \end{equation}
        where \[B_{k,l}\triangleq PA_k^\dag A_l P - \lambda_{k,l}P,\]
         $P$ is the projection on $Q$, and $\ppp{\ket{k}}_{k=0}^{M-1}$ is an orthonormal basis of $\C^M$.
\end{theorem} 

In the following proposition, we modify the assumptions of Theorem~\ref{th:BenyOreshkov} to derive an approximate analog of the KL conditions in \eqref{eq:KLcond} that forms a sufficient
condition for AQEC and is tailored to the specific structure of the codes that we study in Section~\ref{sec:FockCodesConst}.  We remark that analyses in a similar general spirit, though for different settings, 
have earlier appeared in the literature, for instance in \cite{brandao2019quantum,battistel2017general} for Pauli error channels.

\begin{prop}[Approximate KL conditions]\label{prop:KLapprox1}
Let $\calN$ be a quantum channel with Kraus operators $\calE=\ppp{A_k}_{k=0}^{M-1}$, and let $Q\subseteq H$ be a quantum code with an orthonormal basis $\ppp{\ket{c_i}}_{i=0}^{K-1}$ that satisfies 
     \begin{equation}\label{eq: orth}
     \bra{c_i}A_k^\dag A_l \ket{c_j}=0 
     \end{equation}
if $i\neq j$ or $\ell\neq k$. Then $Q$ is an $\varepsilon$-AQECC with respect to $\calE$ if there exist {non-negative} numbers $\ppp{\lambda_{l}}_{l=0}^{M-1}$ such that \(
    \sum_{l=0}^{M-1} \lambda_l=1
\) and
\begin{align*}
    \max_{i,l}\abs{\bra{c_i}A_l^\dag A_l \ket{c_i}-\lambda_{l}}\leq  \frac{\varepsilon^2}{K M}.
\end{align*}

\end{prop}

\begin{proof} Consider the set $\ppp{\lambda_{k,l}}_{k,l}$ such that $\lambda_{k,l}=\delta_{kl}\lambda_l$, where $\lambda_l$ are
as in the above statement, and let $\Lambda$ and $\calB$ be the operators defined in \eqref{eq: Lambda B}. {Note that the matrix defined by $\lambda_{k,l}$ is a density matrix since it is diagonal with positive entries and unit trace by the assumption.} 
By Theorem~\ref{th:BenyOreshkov}, it is sufficient to show that the Bures distance $d(\Lambda,\Lambda+\calB)<\varepsilon$. We will estimate this distance using the diamond norm (see Definition~\ref{def:norms}) and then use \eqref{eq:diamond-Bures} in Lemma~\ref{lem:DiamonaBures}. 
It is thus sufficient to show that  $\norm{\Lambda+\calB-\Lambda}_{\diamond}=\norm{\calB}_{\diamond}<\varepsilon^2$, where $\calB:L(\calH)\to L(\C^M)$ is now given by
    \[
    \calB(\rho)=\sum_{k,l=0}^{M-1} \tr(B_{k,l}\rho)\cdot \ket{k}\bra{l}=\sum_{l=0}^{M-1}\tr(B_{l,l}\rho)\cdot \ket{l}\bra{l}, 
    \]
where the last equality follows since $B_{k,l}=0$ for $k\ne l$ by the assumption \eqref{eq: orth}.

 By Lemma~\ref{lem:HermitianPreserving}, we can write $\norm{\calB}_{\diamond}=\norm{I_{\calH}\otimes \calB (\ket{\psi}\bra{\psi})}_1$ for some $\ket{\psi}\in C^{\otimes 2}$.  By the Schmidt decomposition, we can write $\ket\psi$ as
 $\ket\psi=\sum_{i=0}^{K-1} p_i \ket{i}\otimes \ket{i'}$, where $\sum_ip_i^2=1$ and $\ppp{\ket{i}}_i$, $\ppp{\ket{i'}}_i$ are orthonormal bases of (two copies of) $Q$. Let 
      \[
 \norm{B}_{\s{op}}=\sup_{\norm{X}_1\leq 1}\norm{B(X)}_1
     \]
be the operator norm of $\calB$ with respect to the trace norm.  

With these preparations, we compute
    \begin{align}
     \norm{\calB}_{\diamond}&=\norm{I_{L\p{C}}\otimes \calB (\ket{\psi}\bra{\psi})}_1=\Big\|\sum_{i,j=0}^{K-1} p_ip_j\ket{i}\bra{j}\otimes \calB(\ket{i'}\bra{j'})\Big\|_1 \notag\\
          &\leq \sum_{i,j=0}^{K-1} p_ip_j \norm{\ket{i}\bra{j}}_1\cdot  \norm{\calB(\ket{i'}\bra{j'})}_1 \label{eq: triangle}\\
               &\le\norm{\calB}_{{\s{op}}}\cdot \sum_{i,j=0}^{K-1} p_ip_j
          =\norm{\calB}_{{\s{op}}}\cdot \p{\sum_{i}^K p_i}^2 \label{eq: orthonormal}\\
          &\leq \norm{\calB}_{{\s{op}}}\cdot \p{\sqrt{\sum_{i=0}^{K-1}1^2 \cdot  \sum_{i}^K p_i^2}}^2 \label{eq: Cauchy}\\
     &=K\cdot \sup_{\norm{X}_1\leq 1}\norm{\calB(X)}_1= K\cdot \sup_{\norm{X}_1\leq 1}\norm{\sum_{l=0}^{M-1}\tr(B_{l,l}X)\cdot \ket{i}\bra{i}}_1 \notag\\ 
    &= K\cdot \sup_{\norm{X}_1\leq 1}\sum_{l=0}^{M-1}\abs{\tr(B_{l,l}X) }\leq K\cdot \sum_{l=0}^{M-1}\sup_{\norm{X}_1\leq 1}\abs{\tr(B_{l,l}X) }, \label{eq: sup}
\end{align}
where \eqref{eq: triangle} follows by the triangle inequality and multiplicativity of the trace norm with respect to the tensor product, \eqref{eq: orthonormal} uses the fact that $\norm{\ket{i}\bra{j}}_1=1$ for orthonormal vectors $\ket{i},\ket{j}$
as well as the definition of $\norm{\cdot }_{{\s{op}}}$, and \eqref{eq: Cauchy} follows by the Cauchy–Schwarz inequality.

By the duality principle, Lemma~\ref{lem:Duality}, the supremum in \eqref{eq: sup} can be bounded as 
    \begin{equation}\label{eq: sup B}
    \sup_{\norm{X}_1\leq 1}\abs{\tr(B_{l,l}X)}\leq \norm{B_{l,l}}_\infty, 
    \end{equation}
where $\norm{\cdot}_\infty$ is the usual spectral norm given by the maximal singular value. Recalling that $\ppp{\ket{c_i}}_{i=0}^{K-1}$ is an orthonormal basis, we obtain
   \begin{align*}
    B_{l,l}&=\sum_{i,j=0}^{K-1} \ket{c_i}\bra{c_i}A_l^\dag A_l  \ket{c_j}\bra{c_j} - \lambda_{l,l}\cdot \sum_{i=0}^{K-1} \ket{c_i}\bra{c_i}\\
    &=\sum_{i=0}^{K-1} (\bra{c_i}A_l^\dag A_l \ket{c_i} - \lambda_{l,l})\cdot \ket{c_i}\bra{c_i},
     \end{align*}
where the second equality follows by \eqref{eq: orth}. 
Thus, the set of eigenvalues of $B_{l,l}$ has the form $\{\bra{c_i}A_l^\dag A_l \ket{c_i} - \lambda_{l,l}\}_k$. In particular, 
    \[
    \norm{B_{l,l}}_{\infty}\leq \max_{i}\abs{\bra{c_i}A_l^\dag A_l \ket{c_i} - \lambda_{l,l}}.
    \]
Denote the right-hand side of this inequality by $\varepsilon_{\max}$. Combining this with \eqref{eq: sup} and \eqref{eq: sup B}, we obtain 
\[\norm{\calB}_\diamond\leq K M\varepsilon_{\max}.\]
This completes the proof.
\end{proof}

\section{Bosonic Fock state codes and the AD channel}
We consider a bosonic quantum system comprising $q$ modes, each associated with the separable Hilbert space $\calH_{\mathrm{Fock}}= \ell^2(\N)$, spanned by the number states $\ppp{\ket{n}}_{n\in \Z_0}$.  The joint Hilbert space of the $q$-mode system is 
       \[
       \calH_{q}\triangleq\bigotimes_{i\in [q]}\calH_{\mathrm{Fock}}=\mathrm{Span}\ppp{\ket{\un}=\ket{n_0}\otimes \cdots \otimes \ket{n_{q-1}} ~:~ \un\in \Z_0^q}.
       \]
The total excitation of a number state $\ket{\un}\in \calH_q  $ is given by $N=\sum_{i=0}^{q-1} n_i$. Note that $q$-mode 
number states $\ket{\un}$ with excitation $N$ are in one-to-one correspondence with vertices of the simplex $\calS_{q,N}$, a fact that was used
in the constructions of \cite{aydin2025quantum}. 

We focus on $q$-mode bosonic Fock state codes with constant excitation. As we show in Section~\ref{sec:AmplitudeDamping},  the restriction to constant total excitation both renders the encoding space finite-dimensional and imposes algebraic structure on the action of physically relevant noise operators, which can be exploited in the analysis of quantum error correction. Specifically, a code $Q\subset \calH_q$ of dimension $K$ and constant excitation $N$ is a linear subspace of the space 
    \[
    \calH_{q,N}\triangleq\mathrm{Span}\ppp{\ket{\un} ~:~ \un \in \calS_{q,N}}\subset \calH_q,
    \]
spanned by the basis states
     \[\ket{c_i}=\sum_{\un \in \calS_{q,N}}\alpha_{\un}^{(i)}\cdot \ket{\un}, \quad i=0,\dots,K-1.
     \] 
     The {\em rate} of a $K$-dimensional quantum code $Q$ is defined as 
\[R(Q)\triangleq \frac{\log_2 K}{\log_2 \dim(\calH_{q,N})}=\frac{\log_2K}{\log_2\binom{N+q-1}{q-1}}.\]
where in the second equality we assumed constant excitation.

Although states in $\calH_{q,N}$ have constant total excitation, this condition alone does not preclude the possibility that a large fraction of photons is concentrated in a single mode. To rule out such highly unbalanced patterns, we will limit ourselves to codewords with bounded per-mode occupancy, assuming that every mode cannot host more than a fixed maximum number of photons. This is a simple technical constraint that keeps the code space inside a low-occupation sector and avoids pathological concentration on a few modes. Similar local bounded-occupancy assumptions appear naturally in bosonic encodings used for optical loss tolerance \cite{ralph2005loss}, where logical information is deliberately kept within the vacuum/single-photon sector of each mode. This assumption was earlier used in \cite{grassl2018quantum}, which constructed families of single-error-correcting codes in this regime.
 \begin{definition}
 Let $N,q$, and $B\leq N$ be positive integers. A state in the space $\calH_{q,N}^{(B)}$,
 \[\calH_{q,N}^{(B)}\triangleq\sr{Span}\ppp{\ket{\un} ~:~ \un\in \calS_{q,N}, \norm{\un}_\infty\leq B},\]
 is said to have $B$-local excitation. A Fock state code $Q\subseteq \calH_{q,N}^{(B)}$ is said to have $B$-local excitation. 
 \end{definition}
While this definition implies a deterministic per-mode excitation bound for each codeword, it is also natural to consider relaxed formulations in which local excitation is controlled only in a probabilistic sense with respect to measurement outcomes on a given codeword. Rather than requiring a uniform bound on the maximal photon number per mode, one may instead require that the probability of observing excitation above a prescribed threshold $
B$ under photon-number measurement is small for any input codeword. Such a relaxation allows states whose amplitudes are not strictly confined to a bounded-occupancy sector, while still suppressing most of the time highly concentrated excitation patterns. From this perspective, approximate local boundedness may provide comparable protection against pathological high-excitation configurations, while allowing a larger collection of admissible codes. Accordingly, we give the following definition.
\begin{definition}\label{def:ApproxLocBound}
  Let $N,q$, and $B\leq N$ be positive integers, and let $\Pi_B$ be the orthogonal projection on $\calH_{q,N}^{(B)}$.  A Fock state code $Q\subseteq \calH_{q,N}$ is said to have an $\varepsilon$-approximate $B$-local excitation if for any input codeword, the outcome of the measurement $\ppp{M_0=\Pi_B,M_1=I-\Pi_B}$ gives $0$ with probability at least $1-\varepsilon$, i.e., 
     \[
     \inf_{\ket{\psi_{\sr{in}}}\in Q} \bra{\psi_{\sr{in}}}\Pi_B \ket{\psi_{\sr{in}}}\geq 1-\varepsilon. 
     \]
  
  Let $B=(B_N)_N$ be a sequence of positive numbers.  A sequence of Fock state codes $(Q_N)_N$, $Q_N\subseteq \calH_{q,N}$,  is said to have approximate asymptotically $B$-local excitation if 
      \[
      \inf_{\ket{\psi_{\sr{in}}}\in Q_N} \bra{\psi_{\sr{in}}}\Pi_{B_N} \ket{\psi_{\sr{in}}}\geq 1-o(1).
      \]
 \end{definition}
 
\subsection{The AD channel }\label{sec:AmplitudeDamping}
In this section, we re-examine the formulation of error-correcting conditions for the AD channel accepted in the literature. 

The dominant noise mechanism in bosonic platforms is AD, which models the loss of energy carriers, such as photons or phonons, from the system. The noise strength is characterized by a loss parameter $\gamma\in[0,1]$, typically small, which is an intrinsic property of a single bosonic mode and quantifies its coupling to a loss channel. In multi-mode systems, we assume that $\gamma$  remains a fixed constant for each mode and does not scale with the total number of modes. The AD channel on a $q$-mode bosonic space, $\calN:\calH_q\to \calH_q$, is defined by the set of Kraus operators $\ppp{A_{\ur}}_{\ur\in\Z_0^q}$: 
\[A_{\ur} = A_{r_0}\otimes \cdots \otimes A_{r_{q-1}},\]
where $A_{r_i}$ is the single-mode amplitude error corresponding to $r_i$ photon losses, see \eqref{eq: AD channel Kraus}. The action of $A_{r_i}$ is described as follows:
   \[
  \ket n \;\mapsto 
    \sqrt{\binom{n}{r_i}}(1-\gamma)^{\frac{n-r_i}{2}}\gamma^{\frac{r_i}{2}}\cdot \ket{n-r_i} \Ind_{\{n\ge r_i\}}.
   \]
The error operator $A_{\ur}$ with, $\ur \in \calS_{q,r}$, $r\leq t$, is called an $r$-photon loss error. 

\begin{definition}[Code distance]
A Fock state code $Q$ is said to have distance $d(Q)=t+1$ 
if it can correct any combination of up to $t$ photon-loss events across all modes. In other words, $Q$ is a quantum code that corrects the set of errors $\cE_{\leq t}=\ppp{A_{\ur} ~:~ \ur\in \bigcup_{r\leq t}\calS_{q,r}}$. 
\end{definition}
\begin{tcolorbox}[width=1\linewidth, left=1mm, right=1mm, boxsep=0pt, boxrule=0.5pt, sharp corners=all, colback=white!95!black]  
\begin{definition}[Asymptotically good Fock state codes]
A sequence of codes $\calQ=(Q_N)_N, N\ge 1$, $Q_N\subseteq \calH_{q_{_N},N}$ is called \emph{asymptotically good} if 
    the distance $d(Q_N)$ grows linearly with $N$ and the rate stays positive as $N$ increases:
    \begin{equation}\label{eq:limits}
    \min\Big(R(\calQ),\liminf_{N\to\infty} \frac{d(Q_N)}{N}\Big)>0,
    \end{equation}
where $R(\calQ):=\liminf_{N\to\infty} R(Q_N)$.  
\end{definition}
\end{tcolorbox}

Early work on bosonic Fock state codes under AD noise focused on codes that correct up to a total of 
$t$ photon-loss events. Such codes were introduced in \cite{chuang1997bosonic,leung1997approximate} and subsequently studied in numerous follow-up works, e.g., \cite{grassl2018quantum,OuyangADCode}. A key motivation of this view has been provided in 
\cite{chuang1997bosonic}, which defined the ``code fidelity'' for constant-excitation Fock state codes under $t$ AD errors:
\begin{align*}
 \calF_{\leq t}(Q)\triangleq \inf_{\ket{\psi_{\mathrm{in}}}\in Q} \sum_{\substack{\ur\in \calS_{q,r}\\ r\leq t}} \bra{\psi_{\mathsf{in}}} A_{\ur}^\dag A_{\ur }\ket{\psi_{\mathrm{in}}}.
\end{align*}
As shown in \cite{chuang1997bosonic}, once the excitation of the code is fixed, this quantity does not depend on a particular choice of the code that is
capable of correcting $t$ AD errors. This observation led the authors of \cite{chuang1997bosonic} to suggest the code-fidelity as a universal measure of how well a code correcting $t$ AD errors performs in the presence of the (non-truncated) AD channel (see Sec.~\ref{sec: truncated} for a discussion of the truncation). Relying on the asymptotic behavior of $\calF_{\leq t}$, it was suggested in \cite{chuang1997bosonic} that exact correction of $t$ photon losses might suffice to achieve good approximate error correction against the full AD channel (with a potentially arbitrary number of losses) when $\gamma$ is small. 

A closer look at the code fidelity derivation leads us to a different interpretation of this quantity:
\begin{quote}{\em $\calF_{\leq t}$ is  a universal constant describing the entire constant-excitation subspace $\calS_{q,N}$, and it does not depend on a specific code.}\end{quote}
Thus, the code fidelity is the probability that at most $t$ photons are lost under the AD process for any input state with excitation $N$, and it is independent of the choice of code within that subspace. We also show that a more delicate asymptotic analysis of this quantity is required to determine how well a $t$-error-correcting code actually performs when $N$ increases and $\gamma$ and $t$ change as functions of $N$.
\begin{prop}
    \label{prop:ConstExFidel}
    For any quantum state $\ket{\psi}\in \calH_{q,N}$  we have 
     \[ 
     {\sum_{r\leq t}}\sum_{\substack{\ur\in \calS_{q,r}}} \bra{\psi} A_{\ur}^\dag A_{\ur }\ket{\psi}=
     \sum_{r=0}^t \binom{N}{r}\gamma^{r}(1-\gamma)^{N-r}.
    \]
    \end{prop}
The right-hand side of this relation is the probability that a binomial random variable is at most $t$. Below, we use the notation
    \[
    p_{N,t}:=\P\pp{\sr{Binomial}(N,\gamma)\leq t}.
    \]
Although $p_{N,t}$ depends on $\gamma$, we omit it from the notation since it is fixed throughout all arguments.
\begin{proof}
    Let $\ket{\psi}=\sum_{\un\in \calS_{q,N}}\alpha_{\un}\cdot \ket{\un}$. As shown in \cite[Eq.~(11)]{aydin2025quantum}, for $\ur\in \calS_{q,r}$ we have  
    \[ \bra{\psi} A_{\ur}^\dag A_{\ur }\ket{\psi}= (1-\gamma)^{N-r}\gamma^{r}\binom{N}{r}\binom{r}{\ur}\cdot\sum_{\un\in \calS_{q,N}}\abs{\alpha_{\un}}^2\frac{\binom{N-r}{\un-\ur}}{\binom{N}{\un}}.\]
    Summing over all $\ur\in\bigcup_{r\leq t}\calS_{q,r}$ we obtain 
    \begin{align*}
        \sum_{\substack{\ur\in \calS_{q,r}\\ r\leq t}} \bra{\psi} A_{\ur}^\dag A_{\ur }\ket{\psi}&=\sum_{r=0}^t\sum_{\ur\in\calS_{q,r}}\sum_{\un\in\calS_{q,N}}\abs{\alpha_{\un}}^2(1-\gamma)^{N-r}\gamma^{r}\binom{N}{r}\binom{r}{\ur}\frac{\binom{N-r}{\un-\ur}}{\binom{N}{\un}}
        \\&=\sum_{r=0}^t (1-\gamma)^{N-r}\gamma^{r}  \sum_{\un\in\calS_{q,N}} \abs{\alpha_{\un}}^2 \sum_{\ur\in\calS_{q,r}} \prod_{i=0}^{q-1}\binom{n_i}{r_i}\\
        &=\sum_{r=0}^t (1-\gamma)^{N-r}\gamma^{r}  \binom{N}{r}\sum_{\un\in\calS_{q,N}} \abs{\alpha_{\un}}^2 \\
        &=\sum_{r=0}^t (1-\gamma)^{N-r}\gamma^{r}   \binom{N}{r},
    \end{align*}
    where the second equality follows by a direct calculation (\cite[Lemma A.1]{aydin2025quantum}) and the last equality follows since $\sum_{\un\in\calS_{q,N}} \abs{\alpha_{\un}}^2=1$ as $\ket{\psi}$ is a quantum state.
\end{proof}

Relying on this proposition, we argue that the term ``code fidelity'' is not applicable because $\calF_{\le t}$ does not depend on the code and is simply a binomial probability, $p_{N,t}$, $ t\le N$. In \cite{chuang1997bosonic,wasilewski2007protecting,leung1997approximate,OuyangADCode}, $p_{N,t}$  was approximated as 
   \[
p_{N,t}=1-\binom{N}{t+1}\gamma^{t+1}+O(\gamma^{t+2}).
   \]
This approximation, appropriate for smaller codes, becomes untenable for larger $N$ and $t$. 
For instance, the single-error-correcting code with $N=4$ photons in \cite{chuang1997bosonic} has fidelity $\approx 1 - 6\gamma^2$, which is a good estimate for small $\gamma$. At the same time, the code correcting $t=4$ errors with $N=50$ photons presented in the same paper yields a fidelity bound $\approx 1 - 2\cdot 10^6\gamma^5$ which {becomes vacuous}
for most values of $\gamma\in(0,1)$. 

Turning to asymptotics, we note that in the physically relevant scenario, $\gamma$ is a small constant that does not vanish with $N$, and therefore 
the above approximation is vacuous unless $t$ exceeds the expected number of photon losses $N\gamma$. Specifically, by the weak law of large numbers, $p_{N,t}$ changes from 
almost 0 to almost 1 when $t$ increases within the interval $[(1-\varepsilon)\gamma N, (1+\varepsilon)\gamma N]$ for any $\varepsilon>0$.
 This implies that codes that cannot support correcting a linear number of AD errors will perform poorly on the AD channel with loss parameter $\gamma$ since the probability that the input state will be subjected to a correctable error pattern vanishes with the increase of $N$.
 
 In the next section, we give an operational interpretation of $p_{N,t}$ by relating it to the {\em worst-case entanglement fidelity}, and show that in the asymptotic regime of $N\to\infty,$ $t>(1+\varepsilon)\gamma N$ a code (approximate or exact) that corrects $t$ AD errors, supports approximate correction of the AD noise with respect to the worst-case entanglement fidelity criterion.

\subsection{The truncated AD channel and AQEC}\label{sec: truncated}
As discussed above, given an AD channel, it is not clear what we mean by saying that a code ``corrects $t$ errors''. In this section, we provide a formal
answer to this question.

 Consider the space $\calH_{q,N}$ of constant excitation $q$-mode Fock states and the operation $\calN_{\leq t}$ on $L(\calH_{q,N})$ defined by the normalized truncated AD channel: 
 \begin{align}
     \calN_{\leq t}(\rho)=\frac{1}{p_{N,t}}\sum_{\substack{\ur \in \calS_{q,r}\\r\leq t}} A_{\ur} \rho A_{\ur}^\dag.\label{eq:Trunc<t}
 \end{align}
 
 Proposition~\ref{prop:ConstExFidel} implies that when restricted to $\calH_{q,N}$, $\calN_{\leq t}$, it is trace preserving, and hence a CPTP map. This is formulated in the following proposition, whose formal proof appears in Appendix~\ref{sec:additProof}.
 \begin{prop}\label{obs:quantumChannelTrunc}
     The map $\calN_{\leq t}$ is a quantum channel on $\calH_{q,N}$.
 \end{prop} 

It has been observed since the early days of AQEC that the notion of distance, which plays a central role for exact quantum codes, does not admit a meaningful analog in the approximate setting. In particular, the authors of \cite{crepeau2005approximate} showed that the equivalence between correcting arbitrary errors and erasures, on which the operational interpretation of distance relies, breaks down for approximate quantum error-correcting codes, and concluded that there is no appreciable notion of distance for AQEC. Proposition \ref{obs:quantumChannelTrunc} nevertheless allows us to formulate a notion of asymptotically good approximate Fock state codes by quantifying code performance directly through approximate correction guarantees and bypassing the concept of distance. We formalize this approach in the following definition.

\begin{tcolorbox}[width=1\linewidth, left=1mm, right=1mm, boxsep=0pt, boxrule=0.5pt, sharp corners=all, colback=white!95!black]
\begin{definition}[Asymptotically good approximate Fock state codes]\label{def:  AAG codes} A sequence of AQECCs $\calQ=(Q_N)_N$, $Q_N\subseteq \calH_{q_N,N}$, is called  asymptotically good if there exists $\delta > \gamma >0$ such that $Q_N, N\ge 1$ is an {$\varepsilon_N$}-AQECC for the channel $\calN_{\leq \floor{\delta N}}$, {where $\varepsilon_N\downarrow 0$,} and the sequence $\calQ$ has a positive asymptotic rate, $R(\calQ)>0$. 
\end{definition}
\end{tcolorbox}

We are now ready to give the main statement in this part, showing that the performance of a code correcting $t$ errors on an AD channel is governed by $p_{N,t}$. We remark that while we formulate the statement for  $t$-AQECCs, it also includes exact $t$-error correcting codes, which can be viewed as approximate ones with $\varepsilon=0$.  
\begin{theorem}
\label{th:TruncToAD}
    Assume that $Q\subseteq \calH_{q,N}$ is an $\varepsilon$-AQECC  for the \emph{truncated} channel $\calN_{\leq t}$, then $Q$ is an $\varepsilon'$-AQECC for the AD channel $\calN$, where 
    \[\varepsilon'=\sqrt{1-(1-\varepsilon^2) p_{N,t}}\,.\]
    Conversely, if $Q\subseteq \calH_{q,N}$ is an $\varepsilon$-AQECC for the AD channel $\calN$, then it is also $\varepsilon'$-AQECC for the truncated  channel $\calN_{\leq t}$ with 
    \[\varepsilon'=\frac{\varepsilon}{\sqrt{p_{N,t}}}.\]
\end{theorem}

The proof is long and technical, so it is presented in the appendix. We note that this theorem underscores the operational significance of approximate and exact asymptotically good Fock state codes, as made explicit in the following corollary.

\begin{tcolorbox}[width=1\linewidth, left=1mm, right=1mm, boxsep=0pt, boxrule=0.5pt, sharp corners=all, colback=white!95!black]
\begin{cor}\label{cor:operative}
    A sequence of codes $\calQ=(Q_N)_N$ is \emph{asymptotically good} (in the approximate or exact sense) if and only if it is an AQECC code sequence for the AD channel with a sufficiently small (constant) loss parameter $\gamma>0$. 
\end{cor}
\end{tcolorbox}
\begin{proof}
    Assume that $\calQ$ is an AQECC sequence for the AD channel $\calN$ with $\gamma_0>0$ such that $Q_N, N\ge 1$ is an $\varepsilon_N$-AQECC for $\calN$. Suppose further that the limiting rate of $\cC$, \eqref{eq:limits}, is positive. Fix a number $\delta\in(0,1)$ such that $\gamma_0 < \delta$.
The weak law of large numbers implies that
  \[
  p_{N,\floor{\delta N}}=1-o(1),
  \]
  and by Theorem~\ref{th:TruncToAD}, $Q_N$ is $\varepsilon'$-AQECC for $\calN_{\leq \floor{\delta N} }$ and
$\varepsilon'_N=\varepsilon_N/\sqrt{p_{N,\floor{\delta N}}}=o(1)$, as desired. 

Conversely,  let $\calQ=(Q_N)_N$ be an asymptotically good approximate sequence with respect to $\calN_{\floor{\delta N}}$, with $\delta>\gamma_0>0$ and let $\varepsilon_N\downarrow 0$. Similarly to the previous case, we have $p_{N,\floor{\delta N}}=1-o(1)$, and therefore $Q_N$ is AQECC for $\calN$ with 
    \[
    \varepsilon'_N=\sqrt{1-(1-\varepsilon_N^2)p_{N,t}}=\sqrt{1-(1-o(1))(1-o(1))}=o(1). \qedhere
    \]
\end{proof}

\begin{remark}[Average-case error criterion]
Theorem~\ref{th:TruncToAD} shows that the worst-case AQEC performance against the full AD channel can be controlled in terms of the performance against the truncated channel $\calN_{\leq t}$ and the tail probability $1-p_{N,t}$ of losing more than $t$ photons. The B\'eny--Oreshkov framework is then used to characterize the optimal worst-case recovery error for $\calN_{\leq t}$. Together, these two ingredients yield the bounds on the performance of the code against the full, non-truncated AD channel derived in this subsection.

A parallel analysis can be carried out under the weaker, average-case criterion obtained by replacing the infimum over $\rho$ in \eqref{eq: EF} with the fixed choice $\rho=P_Q/K$ (where $P_Q$ is the projection on $Q$ and $K=\dim(Q)$), or equivalently by evaluating the entanglement fidelity on a maximally entangled state. This is the usual channel-fidelity criterion. The same decomposition of the AD channel into its truncated and tail components yields an analogous relation between the full and truncated channels under the average-case criterion. In this setting, controlling the tail contribution is simpler, since there is no optimization over the input state. The optimal performance on the truncated channel can then be estimated using explicit near-optimal recovery maps, such as Tyson's quadratic recovery~\cite{tyson2010two} or the transpose/Petz recovery~\cite{barnum2002reversing}. These results concern fixed-input entanglement fidelity and, for $\rho=P_Q/K$, specialize to channel fidelity; they therefore provide an average-case counterpart to, rather than an equivalent replacement for, the worst-case B\'eny--Oreshkov criterion employed here. We choose the stronger worst-case criterion in this work, which guarantees uniform performance over all encoded input states. It is nevertheless possible that relaxing to the average-case criterion could yield code families with improved parameters, in particular a better asymptotic rate--loss tradeoff.
\end{remark}


\section{Asymptotically good Fock state codes from \texorpdfstring{$\ell_1$}{} codes}\label{sec:FockCodesConst}

We now turn to the construction of bosonic Fock state quantum codes derived from the classical $\ell_1$ codes we introduced in Section~\ref{sec:ClasicalELl_1}. Our approach follows the general framework of \cite{aydin2025quantum}, which is based on partitioning a classical code into disjoint subsets and associating a quantum basis codeword to each partition element via an appropriately chosen linear combination of the corresponding Fock states. We briefly review this construction and observe that the resulting quantum codes always satisfy the assumptions of Proposition~\ref{prop:KLapprox1}; in particular, they admit reliable identification of the relevant error set.

Using this framework, we obtain three distinct families of asymptotically good Fock state codes derived from the three different $\ell_1$ code constructions presented in Section~\ref{sec:ClasicalELl_1}. The Gilbert-Varshamov (GV) type $\ell_1$ codes lead to Fock state codes achieving exact quantum error correction with the most favorable asymptotic rate–distance tradeoff among the constructions considered. However, this approach relies on computationally involved procedures. One of them is the need to construct a good $\ell_1$ code sequence
given in Proposition~\ref{prop:GVbound}, and the other requires finding a partition of the code into blocks (groups of codewords) such that the resulting Fock state codes satisfy the error correction condition. The complexity of these tasks is at least exponential in the number of modes, rendering the construction infeasible for physically relevant system sizes.

To overcome this limitation, we turn to random $\ell_1$
 codes, which yield families of asymptotically good Fock state codes in the approximate error-correction regime. This construction avoids both computational barriers mentioned above and is therefore significantly simpler from an implementation standpoint. From an operational perspective, this relaxation incurs no essential loss: as established in Corollary~\ref{cor:operative}, both the exact and approximate code families for the truncated amplitude-damping channel imply approximate quantum error correction for the full (and physically more relevant) amplitude-damping channel. Consequently, the choice between these constructions reflects a tradeoff between implementation complexity and asymptotic distance–rate performance, rather than a fundamental degradation in error-correction capability.

 \subsection{The construction of \texorpdfstring{\cite{aydin2025quantum}}{}}
 We start with the definition of 
 classical  $\ell_1$ codes on the simplex, which form the combinatorial backbone of our Fock state code constructions. Define a metric on $\calS_{q,N}$, called the $\ell_1$ distance:
    \[
    d(\ux,\uy)=\norm{\ux-\uy}_1=\frac{1}{2} \sum_{i=0}^{q-1} \abs{x_i-y_i}.
    \]
    A classical code $C\in \calS_{q,N}$ is said to have $\ell_1$ distance at least $t$, $d(C)\ge t$, if every pair of distinct vectors $\ux,\uy\in C$ satisfies $d(\ux,\uy)\geq t$.
    The study of $\ell_1$ codes has a long history in classical coding theory, motivated by applications ranging from constrained storage systems and flash memories \cite{barg2010codes} to DNA storage \cite{yehezkeally2019reconstruction,yehezkeally2021uncertainty} and other communication problems \cite{kovacevic2018multisets,goyal2024gilbert}. Constructions of $\ell_1$ codes often are based on tools from additive number theory, such as the Bose–Chowla theorem and Sidon sets. Here we focus on random coding and geometric GV-type arguments to show the existence of $\ell_1$ codes with favorable rate–distance tradeoffs beyond what is achievable by algebraic constructions. 

A general approach to the construction of Fock state codes from classical $\ell_1$ codes was suggested in \cite{aydin2025quantum}.

 \begin{tcolorbox}[width=1\linewidth, left=1mm, right=1mm, boxsep=0pt, boxrule=0.5pt, sharp corners=all, colback=white!95!black] 
\begin{construction}\label{construction: AAB} Let $C_N\subseteq \calS_{q,N}$ be a code with $\ell_1$ distance $t+1$. Given a partition of $C_N$ into disjoint subsets $I_0,I_1,\dots,I_K$, consider the Fock state code $Q_N$ with the basis 
 $\{\ket{c_0},\dots,\ket{c_{K-1}}\}\subset \calH_{q,N}$ where 
 \begin{align}
     \ket{c_i}= \sum_{\un\in I_i}\alpha_{\un}^{(i)
 } \ket{\un},\label{eq:CodeBasis}
 \end{align}
and $\set{\alpha_{\un}^{(i)}}_{\un \in I_i }$ is a set of coefficients such that $\sum_{\un \in I_i}|\alpha_{\un}^{(i)}|^2=1$.
\end{construction}
\end{tcolorbox}
\noindent We start by observing that $Q_{N}$ always satisfies the KL orthogonality condition. The only remaining constraint to be addressed for constructing exact or approximate quantum error-correcting codes is the non-deformation condition. Although this fact is established as part of the proof of Theorem~7 in \cite{aydin2025quantum}, we briefly outline the main idea and use it to show that such a code always admits error identification, in an even more general setting.
 \begin{prop} \label{obs:KLorth}For $i\ne j$ and $\ur\ne \ur'$ such that $\norm{\ur}_1,\norm{\ur'}_1\leq t$, 
     the basis codewords defined in \eqref{eq:CodeBasis} satisfy 
     \begin{align}
         \bra{c_i}A_{\ur}^\dag A_{\ur'} \ket{c_j}=0. \label{eq:condOrth}
     \end{align}
 \end{prop}
 \begin{proof}
     First, note that $A_{\ur}$ maps $\ket{c_i}$ to $\calH_{q,N-r}$, and since $\calH_{q,N-r} $ and $ \calH_{q,N-r'} $ are orthogonal to one another, we may assume that $\ur,\ur'\in \calS_{q,r}$ for some $r<N$. A simple calculation (see \cite[pp. 8]{aydin2025quantum}) shows that 
     \begin{align*}
         \bra{c_i}A_{\ur}^\dag A_{\ur'} \ket{c_j}=\sum_{\substack{\un\in I_i\\ \un'\in I_j}}\gamma^{r}(1-\gamma)^{N-r}\binom{N}{r}\sqrt{\frac{\binom{r}{\ur}\binom{r'}{\ur}\binom{N-r}{\un -\ur}\binom{N-r}{\un' -\ur'}}{\binom{N}{\un }\binom{N}{\un'}}}(\alpha_{\un}^{(i)})^* \alpha_{\un}^{(j)}\cdot \braket{\un-\ur | \un'-\ur'}.
     \end{align*}
Suppose that $\braket{\un-\ur | \un'-\ur'}\ne 0$ for some $\un,\un',\ur,\ur'$, which then should satisfy $\un-\ur=\un'-\ur'$.
 First, suppose that $i\ne j$, then $I_i\cap I_j=\emptyset$ by assumption. This implies that $\un\neq \un'$ and therefore, $d(\un,\un')\geq t$ by the assumption on the distance of $C_N$. On the other hand, $\norm{\ur-\ur'}_1\leq 2r$, implying that
 $d(\un-\ur,\un'-\ur')\ge 1$, so the equality cannot hold. Now let $i=j$, in which case $\un$ and $\un'$ can be distinct or coincide. If they are distinct, then the equality cannot be true by the same argument as above. If $\un= \un'$, then $\un-\ur\ne \un'-\ur'$ by our assumption of $\ur\ne \ur'$. We conclude that $\braket{\un-\ur | \un'-\ur'}=0$ for each term in the sum, and so \eqref{eq:condOrth} holds.
 \end{proof} 

Below, we argue that the codes $Q_N$ can identify up to $t$ errors when used on the AD channel. To formalize this claim, 
we recall that with probability $p_{N,t}$, the action of the channel is described by one of the operators $A_r$ and that with 
probability $1-p_{N,t}$ the input state $\ket{\psi_{\sr{in}}}$ is subjected to the loss of $>t$ photons. An error resulting in a loss of at most $t$ photons can be identified if there exists a projective measurement $\ppp{\Pi_{\ur}:\ur\in \calS_{q,r}, r\leq t}\cup \ppp{\Pi_{>t}}$ such that once $\ket{\psi_{\sr{in}}}$ is submitted to the channel, the outcome $\ur$ is obtained if the channel acts as $A_{\ur}$, while the outcome $>t$ is obtained with probability $1-p_{N,t}$. Furthermore, this measurement remains valid even if $I_0,\dots, I_K$ are not disjoint, meaning that the code $Q_N$ is any subspace spanned by Fock states corresponding to classical codewords from $C_N$. That is formulated in Corollary~\ref{cor:identification}, whose proof appears in Appendix~\ref{sec:additProof}.
 
 \begin{cor}\label{cor:identification}
 The code $Q_N=\langle \ket{c_i}, i\in[K]\rangle$ can identify all errors caused by at most $t$ photon losses.
 \end{cor}

 \begin{remark} Quantum error identification, considered here, stops short of correcting the errors, but adds 
functionality compared to other noise-resistant primitives, such as, for instance, error detection, which
plays an essential role in fault tolerance. For instance, 
 in post-selected fault-tolerant architectures and in magic-state distillation protocols, error-detecting codes are used to reject faulty computational branches or ancilla states rather than to actively correct them, enabling arbitrarily low logical error rates through filtering and purification instead of recovery \cite{knill2005quantum,bravyi2005universal}. Error identification goes further by providing classical information about the specific error instance. This additional information enables adaptive feedback and real-time correction strategies, such as photon-loss tracking in bosonic codes \cite{ofek2016extending}, and underlies decoding procedures in topological codes where syndrome patterns are mapped to specific error configurations before correction \cite{fowler2012surface}. 
 \end{remark}

 \subsection{Asymptotically good Fock state codes}\label{sec:random}
 In this section, we construct two families of asymptotically good approximate Fock state codes using the random $\ell_1$
 codes from Section 3 within the framework of Construction~\ref{construction: AAB}. While orthogonality follows from the $\ell_1$ distance guarantees, randomness is used to ensure the non-deformation KL conditions, which are shown to hold with high probability under a suitable averaging over partitions. We show that the obtained Fock state codes inherit the local boundedness properties from the corresponding random classical codes, which guarantees the approximately bounded local excitation, cf. Definition~\ref{def:ApproxLocBound}. 

We will construct Fock state codes from classical $\ell_1$ codes as follows. Let $C \subset \calS_{\alpha N, N}$ be a code of size $L=TK,$ where $K$ is the target dimension of the quantum code\footnote{if $K\hspace*{-2mm}\not| L$, we can discard $L\, \text{mod}K$ points of the code to satisfy the divisibility condition.}. Let $I_0,\dots,I_{K-1}$ be an arbitrary partition of $C=\{\un^{(0)},\dots,\un^{(L-1)}\}$ into equal-size parts. Consider the Fock state code $Q_N$ defined in \eqref{eq:CodeBasis}, where the coefficients of the basis states are given by 
    \begin{align}
        \alpha_{\un}^{(j)}=\Ind_{I_j}(\un) \frac{1}{\sqrt{T}}, \quad j\in [K],\label{eq:randCodeCoeff}
    \end{align}
where $\Ind_{I_j}$ is the indicator of the set $I_j$. 
 
 \begin{prop}\label{Prop:FockByRand}
 Fix $\alpha >0$, $\gamma\in (0,1]$, and $0<\delta<1$. Let $C_N=\set{\su{n}^{(i)}}_{i=0}^{L}$  be an $\ell_1$ code, whose codewords are drawn independently from $\calS_{\alpha N,N}$ according to some distribution $\s{D}_N$ on the simplex. Suppose that 
 $\frac{\ln N}{\ln L}\to 0$ as $N\to\infty$. Suppose further that
 the sequence of Fock state codes
 $(Q_N)_N$ constructed as in \eqref{eq:randCodeCoeff} is used on the truncated AD channel $\calN_{\leq \floor{\delta N}}$ defined in \eqref{eq:Trunc<t} {by the Kraus operators $\Tilde{A}_{\ur}=A_{\ur}/p_{N,\floor{\delta N}}$, $\ur\in \calS_{\alpha N,r}$, $r\leq \floor{\delta N}$.} Let $M$ be the number of Kraus operators of this channel\,\footnote{We again suppress the dependence of the parameters $K, L,$ and $M$ on $N$.}.

 There is a choice of {non-negative constants} $\lambda_{\ur}$, $\ur\in \calS_{\alpha N,r}$, $r\leq \floor{\delta N}$ 
 that sum to 1 (they may depend on $N$) and an  $o(1)$ function, such that with probability approaching 1 as $N$ increases, the sequence $(C_N)_N$ satisfies 
    \begin{align}
        \max\ppp{\abs{\bra{c_i} {\Tilde{A}_{\ur}^\dag \Tilde{A}_{\ur}} \ket{c_i}-\lambda_{\ur}} ~:~ i\in [K], \ur\in \calS_{\alpha N,r}, r\leq \floor{\delta N}  }=o\Big(\frac{1}{K  M }\Big)\label{eq:non-deformation}
    \end{align}
 as long as 
    \begin{align}
         K^3 M^2 \leq L^{1-\varepsilon} ,\label{eq:KcondRand}
    \end{align}
    for some $\varepsilon>0$. 
 \end{prop}

 \begin{proof}
  Denote $t=\floor{\delta N}$. Our goal is to find {non-negative numbers} $\set{\lambda_{\ur}}_{\ur \in \calS_{N,r},r\leq t}$ {that sum to 1} and $\xi_N>0$ such that 
    \[\P\pp{\max_{\ur\in \cup_{r\leq t}\calS_{\alpha N,r}}\max_{
        j=0,\dots,K-1}\abs{\bra{c_j}{\Tilde{A}_{\ur}^\dag \Tilde{A}_{\ur}} \ket{c_j}-\lambda_{\ur}}>\xi_N}\xrightarrow[N\to\infty]{}0,\]
        where $\xi_N=o\p{\frac{1}{KM}}$.
     To this end, for every $\un\in \calS_{\alpha N,N}$ and $\ur\in \calS_{\alpha N,r}$ define:
    \[
    Y_{\ur}(\un):={\frac{1}{p_{N,t}}}(1-\gamma)^{N-r}\gamma^r \cdot \prod_{i=0}^{q-1} \binom{n_i}{r_i}\leq {\frac{1}{p_{N,t}}}(1-\gamma)^{N-r}\gamma^r \binom{N}{r}.
    \] 
    {Recall that $p_{N,t}$ is the probability that a random variable $\s{Z}\sim\s{Bin}(N,\gamma)$  is bounded by $t$. Since $r\leq t$, we find that 
    \begin{equation}
        Y_{\ur}(\un)\leq \P[Z=r |Z\leq t]< 1.\label{eq: Yr}
    \end{equation}}
for all $\un\in \calS_{\alpha N,N}$ and $\ur\in \calS_{\alpha N,r}$.

As in Prop. \ref{prop:ConstExFidel} above, we use \cite[Eq.(11)]{aydin2025quantum} to claim that for $\ur\in\calS_{N,r}$, we have 
    \begin{align*}
        \bra{c_j}{\Tilde{A}_{\ur}^\dag \Tilde{A}_{\ur}} \ket{c_j}&={\frac{1}{p_{N,t}}}(1-\gamma)^{N-r}\gamma^r\binom{N}{r}\binom{r}{\ur} \sum_{\un\in I_j}\abs{\alpha^{(j)}_{\un}}^2  \frac{\binom{N-r}{\un-\ur}}{\binom{N}{\un}}\\
        &={\frac{1}{p_{N,t}}}(1-\gamma)^{N-r}\gamma^r \frac{N!}{r!(N-r)!} \frac{r!}{\prod_{i=0}^{q-1} r_i} \sum_{\su{n}\in I_j} \frac{1}{T}\frac{(N-r)!}{N!}\prod_{i=0}^{q-1} \frac{\s{n}_i!}{(\s{n}_i-r_i)!}\\
        &={\frac{1}{p_{N,t}}}(1-\gamma)^{N-r}\gamma^r\sum_{\su{n}\in I_j} \frac{1}{T}\prod_{i=0}^{q-1} \binom{\s{n}_i}{r_i}\\
        &=\frac{1}{T}\sum_{\su{n}\in I_j}Y_{\ur}(\su{n}).
    \end{align*}
    Define 
    \[
    \lambda_{\ur} ={\frac{1}{p_{N,t}}}(1-\gamma)^{N-r}\gamma^r \;\E_{\su{X}\sim \s{D}_N}\Big[  \prod_{i=0}^{q-1} \binom{\s{X}_i}{r_i}\Big]{\geq 0}.
    \]
    By Hoeffding's inequality \eqref{eq: Hoeffding} and \eqref{eq: Yr}, for any $\xi>0$ we have
    \begin{align}
        \P\Big[\Big|\frac{1}{T}\sum_{\su{n}\in I_j}Y_{\ur}(\su{n})-\lambda_{\ur}\Big|>\xi \Big\}\leq2\exp\p{{-2 T\xi^2}}.\label{eq:probBoundK^3}
    \end{align}
    Setting {$\xi_N=\frac{1}{L^{\varepsilon/4} K M}$ } and recalling that $L=TK$, we obtain
    \begin{align}
         \nonumber\P\pp{\abs{\bra{c_j}{\Tilde{A}_{\ur}^\dag \Tilde{A}_{\ur}} \ket{c_j}-\lambda_{\ur}}>\xi_N}&=
         \P\Big[\Big|\frac{1}{T}\sum_{\su{n}\in I_j}Y_{\ur}(\su{n})-\lambda_{\ur}\Big|>\frac{1}{{L^{\varepsilon/4}} K M}\Big]\\
         &\leq 2\exp\p{-\frac{2L  }{{L^{\varepsilon/2} }K^3M^2}}\notag\\
         &\leq \exp(-L^{{\varepsilon/2}-o(1)}),\label{eq:ProbInnerBound}
    \end{align}
    where the last equality follows from the assumption in \eqref{eq:KcondRand}.
With this, we finish the proof as follows:
    \begin{align*}
        \P\Big[\max_{\ur\in \cup_{r\leq t}\calS_{N,r}}&\max_{
        j=0,\dots,K-1}\abs{\bra{c_j}{\Tilde{A}_{\ur}^\dag \Tilde{A}_{\ur}} \ket{c_j}-\lambda_{\ur}}>\xi_N\Big]\\ 
        & =\P\Big[\bigcup_{\ur\in \cup_{r\leq t}\calS_{N,r}}\bigcup_{
        j=0,\dots,K-1}\ppp{\abs{\bra{c_j}{\Tilde{A}_{\ur}^\dag \Tilde{A}_{\ur}} \ket{c_j}-\lambda_{\ur}}>\xi_N}\Big]\\ 
        & \leq \sum_{\ur\in \cup_{r\leq t}\calS_{N,r}} \sum_{j=0}^{K-1} \cdot   \P\pp{\abs{\bra{c_j}{\Tilde{A}_{\ur}^\dag \Tilde{A}_{\ur}} \ket{c_j}-\lambda_{\ur}}>\xi_N}\\ 
        & \leq \Big|\bigcup_{r\leq t}\calS_{N,t}\Big| K\cdot   e^{-L^{{\varepsilon/2}+o(1)}}
       \\ 
        & \leq 2^{4N} \cdot L\cdot  e^{-L^{{\varepsilon/2}+o(1)}}\\&\leq  e^{-L^{{\varepsilon/2}-o(1)}},
    \end{align*}
which goes to 0 as $N$ increases. On the last line, we used the limiting condition $N/L^\epsilon\to 0$, implied by our
assumption of $\frac{\ln N}{\ln L}\to 0$.

{To complete the proof, it remains to show that indeed $\sum_{\ur}\lambda_{\ur}=1$. Indeed, 
\begin{align*}
    \sum_{r\leq t}\sum_{\ur\in \calS_{\alpha N, r} }\lambda_{\ur}&=\sum_{r\leq t}\sum_{\ur\in \calS_{\alpha N, r}}\frac{1}{p_{N,t}}(1-\gamma)^{N-r}\gamma^r \;\E_{\su{X}\sim \s{D}_N}\Big[  \prod_{i=0}^{q-1} \binom{\s{X}_i}{r_i}\Big]\\
    &=\sum_{r\leq t}\frac{1}{p_{N,t}}(1-\gamma)^{N-r}\gamma^r \;\E_{\su{X}\sim \s{D}_N}\Bigg[  \sum_{\ur\in \calS_{\alpha N, r}} \prod_{i=0}^{q-1} \binom{\s{X}_i}{r_i}\Bigg]
    \\&=\sum_{r\leq t}\frac{1}{p_{N,t}}(1-\gamma)^{N-r}\gamma^r \;\E_{\su{X}\sim \s{D}_N}\Big[ \binom{\sum_{i=0}^{q-1} \s{X_i}}{ r} \Big]\\
    &=\sum_{r\leq t}\frac{1}{p_{N,t}}(1-\gamma)^{N-r}\gamma^r \binom{N}{r }=1,
\end{align*}
where the equality at the last line follows since under the distribution $\s{D}_N$, with probability $1$, we have $\su{X}\in\calS_{\alpha N, N}$.
}
   \end{proof}

 We are now ready to state the main result of this section, which gives the parameters of asymptotically good AQECCs that are obtained from our construction.
\begin{tcolorbox}[width=1\linewidth, left=1mm, right=1mm, boxsep=0pt, boxrule=0.5pt, sharp corners=all, colback=white!95!black]
 \begin{theorem}
     Fix $\alpha >0$, $0<\gamma,\delta<1$ and $\epsilon>0$. Let $C_N=\set{\su{n}^{(i)}}_{i=0}^{L}$  be an $\ell_1$ code, whose codewords are drawn independently from $\calS_{\alpha N,N}$ according to some distribution $\s{D}_N$ on the simplex.
 Suppose that
 the sequence of ${K_N}$-dimensional Fock state codes
 $(Q_N)_N$ constructed as in \eqref{eq:randCodeCoeff} is used on the truncated AD channel $\calN_{\leq \floor{\alpha N}}$ defined in \eqref{eq:Trunc<t}. If 
\begin{align}
         0&<\limsup_{N\to\infty }\frac{\log_2 L}{N}= \bar R(\delta)-o(1) \notag\\
         \frac{\log_2 {K_N}}{N}& <\frac13 \bar R(\delta)-\frac{2}{3}(1+\delta)h_2\p{\frac{1}{1+\delta}}-\varepsilon,\label{eq:CondUnifRates}
         \end{align}
then with probability $1-o(1)$ the codes $(Q_N)_N$ are AQECCs for the channel sequence $(\calN_{\leq \floor{\alpha N}})_N$ and they also have approximate asymptotically $B$-local excitation (cf. Def.~\ref{def:ApproxLocBound}). Here
      \[
          (\bar R(\delta),B)=\begin{cases}( R_{\sr{U}}(\delta),(1+\varepsilon)\log_{1+\alpha}N) &\text{if $\s{D}_N={\sf{Unif}}(\calS_{\alpha N,N})$}\\[.1in]
          (R_{\sr{M}}(\delta),(1+\varepsilon) \frac{\ln N}{\ln\ln N})&\text{if $\s{D}_N=\s{Multinomial}(\floor{\alpha N},N)$}.
          \end{cases}
      \]
  \end{theorem}
  \end{tcolorbox}
For instance, for $\alpha=5$, the rate of the quantum codes constructed in this theorem is positive for $\delta\le .15$ for the uniform ensemble and $\delta\le 0.05$ for the multinomial ensemble.    
 
 \begin{proof} Throughout this proof, we write $K$ instead of $K_N$.
Both cases of the claim follow the same argument, so we only verify the case of the uniform distribution. The number of Kraus operators of the channel $\calN_{\leq \floor{\delta N}}$ is given by 
     \begin{align}
        M=\sum_{i=0}^t \abs{\calS_{N,r}}=\sum_{i=0}^t \binom{r+N-1}{N-1}\leq 2^{N(1+\delta)(h_2(\frac{1}{1+\delta})+o(1))},\label{eq:NumKraus}
    \end{align}
Coupled with \eqref{eq:CondUnifRates}, we conclude that
    \begin{align}
        K^3\leq \frac{L^{1-\varepsilon}}{M^2}, \quad \frac{\ln N}{\ln L}=o(1).
        \label{eq:CondSat}
    \end{align}
Thus,  by Propositions~\ref{prop:randEll_1}, \ref{Prop:FockByRand}, and \ref{obs:KLorth}, with probability $1-o(1)$, the Fock state code sequence $(Q_N)_N$ and the underlying random $\ell_1$ code sequence $(C_N)_N$ have the following properties:
     \begin{enumerate}
         \item The $\ell_1$ distance of $C_N$ is at least $\floor{\delta N}$, and therefore the corresponding Fock state code $Q_N$ satisfies \eqref{eq:condOrth}.
         \item There are constants $\set{\lambda_{\ur}}$ such that the approximate non-deformation condition~\eqref{eq:non-deformation} holds.
     \end{enumerate}
     Using Proposition~\ref{prop:KLapprox1}, we conclude with probability $1-o(1)$,  the sequence $(Q_N)_N$ is $\varepsilon_N$-AQECC for the channels $\calN_{\leq \floor{\delta N}}$ 
     for some $\epsilon_N\downarrow 0$. 
     
 It remains to prove the claim about the bounded local excitation. Let $\ket{\psi_{\sr{in}}}\in Q_N$ be a codeword and let $\Pi$ be the orthogonal projection on $\calH_{\floor{\alpha N},N }^{(B_N)}$, where $B_N=\p{1+\varepsilon}\log_{1+\alpha}(N)$. Let us write $\ket{\psi_{\sr{in}}}=\sum_{i=0}^{k-1} \beta_i \ket{c_i}$, where $(\ket{c_i})$ is the basis defined in \eqref{eq:CodeBasis}, \eqref{eq:randCodeCoeff}. A simple calculation shows that 
     \begin{align}
         \bra{\psi_{\sr{in}}} \Pi \ket{\psi_{\sr{in}}}&=\sum_{i,j=0}^{K-1} \beta_i^*\beta_j \bra{c_i}\Pi \ket{c_j}=\sum_{i,j=0}^{K-1}  \frac{K}{L}\cdot \beta_i^*\beta_j \sum_{\un\in I_i} \sum_{\un'\in I_j } \bra{\un}\Pi \ket{\un'}. \label{eq:projectCalc}
     \end{align}
To estimate this expression, observe that for any $\un, \un'\in \calS_{\floor{\alpha N},N}$, we have 
 $\bra{\un}\Pi \ket{\un'}= \delta_{\un,\un'}\Ind_{\ppp{\norm{\un}_\infty\leq B_N}}$, and therefore,
     \begin{align}
         \bra{\psi_{\sr{in}}} \Pi \ket{\psi_{\sr{in}}}\nonumber&=\sum_{i,j=0}^{K-1}  \frac{K}{L} \beta_i^*\beta_j \sum_{\un\in I_i} \sum_{\un'\in I_j } \delta_{\un,\un'} \Ind_{\ppp{\norm{\un}_\infty}\leq B_N}\\
         \nonumber&=\sum_{i=0}^{K-1}  \frac{K}{L} |\beta_i|^2 \abs{\ppp{\un \in I_i ~:~ \norm{\un}_{\infty}\leq B_N}}\\
         &=\sum_{i=0}^{K-1}  |\beta_i|^2 \Big(1- \frac{K}{L}\cdot \abs{\ppp{\un \in I_i ~:~ \norm{\un}_{\infty}> B_N}}\Big)\\
         &\geq 1-\max_{i\in [K]}\frac{K}{L}\cdot \abs{\ppp{\un \in I_i ~:~ \norm{\un}_{\infty}> B_N}},\label{eq:psi_inEval}
     \end{align}
     where we used the fact that $I_i\cap I_j=\emptyset$ and that $\sum_i |\beta_i|^2=1$.  Recall that $I_i$ is a partition block formed of $T=L/K$ i.i.d. random vectors distributed uniformly on the simplex. Let $\s{Y}_i$ be the random variable that counts the fraction of `bad' words in $I_i$:
     \[  \s{Y}_i=\frac{1}{T}\cdot \abs{\ppp{\un \in I_i ~:~ \norm{\un}_{\infty}> B_N}}. \]
     By \eqref{eq:psi_inEval}, to show that $\min_{\ket{\psi_{\sr{in}}}}\bra{\psi_{\sr{in}}} \Pi \ket{\psi_{\sr{in}}}=1-o(1)$ with high probability, it is sufficient to show that 
     \begin{align}
         \P\pp{\max_{i\in [K]} \s{Y}_i \leq g(N)} \xrightarrow[N\to \infty]{} 1,\label{eq:YcondG}
     \end{align}
     where $g(N)\to 0$. To that end, let $p_N$ be the probability that a uniform vector on the simplex $\su{X}\sim \sr{Unif}(\calS_{\alpha N , N })$  satisfies $\norm{\su{X}}_\infty >B_N=\p{1+\varepsilon} \log_{1+\alpha} (N)$. By \eqref{eq:InftyNormBoundProb} and the definition of the random $\ell_1$ codeword $\su{n}\in I_i$
     \begin{align*}
          p_N:=\E\big[\Ind_{\ppp{\norm{\su{n}}_{\infty}>B_N}}\big]=\P\pp{\norm{\su{X}}_{\infty}>B_N}=O(N^{\frac{\varepsilon^2-\varepsilon}{2}}).
     \end{align*}
     For  some $0<\xi<\frac{1}{2}$, let $g(N)=p_N+T^{-\xi}=o(1)$. Using Hoeffding's inequality \eqref{eq: Hoeffding}, we have that 
     \begin{align*}
        \P\pp{\s{Y}_i>g(N)}&=\P\Big[\frac{1}{T}\cdot \sum_{\su{n}\in I_i}\p{\Ind_{\ppp{\norm{\un}_{\infty}>B_N}}-\E\pp{\Ind_{\ppp{\norm{\un}_{\infty}<B_N}}}}>T^{-\xi}\Big]\\
         &\leq 2\exp\p{-\frac{2T^{2-2\xi}}{T}}=2e^{-2T^{1-2\xi}}.
     \end{align*}
     We are now ready to make the final step: 
     \begin{align*}
          \P[\max_{i\in [K]} \s{Y}_i > g(N)]&\leq  \sum_{i\in [K]}\P\pp{\s{Y}_i > g(N)}\leq 2K e^{-2T^{1-2\xi}} \\
          &=2\exp\Big(-\p{\frac{L}{K}}^{1-2\xi}+\log K\Big).
          \end{align*}
Rewriting \eqref{eq:CondSat}, we obtain $K<\frac{L^{1/3-\varepsilon}}{M^{2/3}}<L^{1/3}$. Using this on the previous line, 
          \begin{align*}
           \P[\max_{i\in [K]} \s{Y}_i > g(N)]&\leq 2\exp\Big(-L^{\p{\frac{2}{3}(1-2\xi)}}+\frac{1}{3}\log L\Big).
           \end{align*}
     Since the right-hand side vanishes as $N\to\infty$, we have proved \eqref{eq:YcondG} and therefore, the theorem.
 \end{proof}

\begin{remark}{[{\sf Discussion}]}\label{remark: Tverberg}
The codes defined in \eqref{eq:CodeBasis}, \eqref{eq:randCodeCoeff}, which we employ in this section,  
rely on an arbitrary partition of the classical code into equal-sized blocks,   
followed by averaging of the associated quantum states within each subset, yielding approximate non-deformation KL conditions. In \cite{aydin2025quantum} we used partitions arising from Tverberg's theorem to prove validity of the KL conditions (cf. Theorem~\ref{th:AABmain}). Here, condition \eqref{eq:KcondRand} replaces the use of Tverberg-type arguments, supporting exact non-deformation conditions. 
    
There is also another relaxation of the convex-geometric arguments, known as the no-dimen\-sional (approximate) Tverberg theorem \cite{adiprasito2020theorems}, which asserts that any set of $L$ points contained in the Euclidean ball of diameter $D$ can be partitioned into  $K$ subsets whose convex hulls all intersect a common ball of radius $O(\sqrt{K/L}) D$. Interestingly, this non-constructive geometric argument yields exactly the same scaling relation between $K,L,$ and $M$ as \eqref{eq:KcondRand}. It is also known that there exists a procedure that finds such partitions as well as the common ball in time that is polynomial with respect to the number of points and the ambient dimension. 
    
However, in our setting, this geometric approach has two important drawbacks. First, when instantiated for Fock state codes, the 
ambient Euclidean dimension corresponds to the number of modes, which grows exponentially with the number of physical qudits. As a 
result, even algorithms that are polynomial in the dimension become exponential when the complexity is measured by the relevant physical parameters. Second, 
the resulting convex combinations are not guaranteed to preserve the required local boundedness (per-mode excitation) constraints, 
which are essential for the implementation of the code. For these reasons, even though both approaches apply
under the same set of assumptions on the parameters, we adopt the uniform partition and averaging construction. 
\end{remark}
 
 \subsection{De-randomized asymptotically good Fock state codes}
In this section, we present a deterministic construction of asymptotically good Fock-state codes, obtained by derandomizing the uniform random $\ell_1$ simplex code via a Gilbert--Varshamov-type argument (see Section~\ref{sec:GV}) and then reverting to the formalism of \cite{aydin2025quantum}. Although this approach entails high complexity of construction similar to \cite{aydin2025quantum}, we include it to show that, at least theoretically, a more favorable error-correction/distance tradeoff is possible. As discussed above, the improvement of parameters over the randomized approach obtained here stems from the fact that Tverberg partitions result in less stringent constraints on the code dimension than those arising from the more general concentration arguments used in the random construction.
 \begin{theorem}{\rm (\cite[Theorem~7]{aydin2025quantum})} \label{th:AABmain} Let $C_N\subseteq \calS_{q,N}$ be an $\ell_1$ code  with distance at least $t+1$. If  
 \[\abs{C_N}\geq \p{K-1}\binom{t+q-1}{t-1}+1.\]
 then there exist a partition $C_N=\sqcup_{i=0}^{K-1} I_i$ and a choice of coefficients $\set{\alpha_{\un}^{(i)}}$ such that the corresponding Fock state  code from Construction \ref{construction: AAB} has dimension $K$ and distance $t+1$.
  \end{theorem}
 We use Theorem~\ref{th:AABmain} and the GV-type bounds presented in Section~\ref{sec:ClasicalELl_1} to construct exact asymptotically good Fock state codes with the additional property of logarithmic per-mode excitation, inherited from the structure of the classical codes that underlie the construction. 
\begin{tcolorbox}[width=1\linewidth, left=1mm, right=1mm, boxsep=0pt, boxrule=0.5pt, sharp corners=all, colback=white!95!black]
 \begin{theorem}[Exact asymptotically good Fock state codes]\label{th:ExactAG} For any $0<\delta<1$, $\varepsilon>0$ and $\alpha>0$ there exists a sequence of Fock state codes $(Q_N)_N$, $Q_N\subseteq\calH_{\alpha N, N}$ with distance at least $\floor{\delta N}$, local excitation $B=(1+\varepsilon)\log_{1+\alpha} N$ and dimension ${K_N}$  as long as 
     \begin{equation}\label{eq: rate Fock}
 \frac{1}{N}\log_2 {K_N}\leq  R_{\sr{GV}}(\delta,\alpha)-(\alpha+\delta) h_{2}\p{\frac{\alpha}{\alpha+\delta}}+o(1),
    \end{equation}
 for some fixed $o(1)$ function.
 \end{theorem}
 \end{tcolorbox}
 \begin{proof}
    Consider the sequence of GV-type codes $C_N$ constructed in Proposition~\ref{prop:GVbound}. By Theorem~\ref{th:AABmain}, there exists a partition $I_0,\dots,I_{K-1}$ of $C_N$ and a choice of coefficients $\set{\alpha_{\un}^{(i)}}$ such that the corresponding Fock state  code obtained from Construction~\ref{construction: AAB} has dimension $K$ and distance $\floor{\delta N}+1$ given that 
    \[K_N<\abs{C_N} \binom{\floor{\delta N}+\floor{\alpha N}-1}{\floor{\alpha N}-1}^{-1}.\]
For $N\to\infty$,
    \begin{align*}
       \log_2K_N&<\log_2\abs{{C}_N}-\log_2\binom{\floor{\delta N}+\floor{\alpha N}-1}{\floor{\alpha N}-1}\\
        &=N(R_{\sr{GV}}(\delta,\alpha)-(\alpha+\delta)h_2\p{\frac{\alpha}{\alpha+\delta}}+o(1)),
    \end{align*}
where the last equality follows from Proposition~\ref{prop:GVbound} and \eqref{eq:asymbinom}. Finally, the fact that $Q_N$ has local excitation $B=(1+\varepsilon)\log_{1+\alpha}{N}$  follows immediately since for any classical codeword $\un$, it is guaranteed by Proposition~\ref{prop:GVbound} that $\norm{\un}_\infty\leq B$.
      \end{proof}
This theorem implies the existence of codes of positive rate for any $\delta>0$ such that 
the right-hand side of \eqref{eq: rate Fock} is positive. To argue that this region is non-vacuous, note that it is positive for 
any $\alpha>0$ and $\delta=0$, and so there exists a neighborhood of $\delta=0$ for which 
there exist infinite sequences of codes of distance $\delta N$ and strictly positive rate.
\section{Asymptotically good PI and spin codes}\label{sec: PI Spin}
\subsection{PI codes}
Our construction of asymptotically good Fock state codes implies existence of another asymptotically good family of quantum codes called \textit{permutation-invariant  codes}. PI codes, which lie in the symmetric subspaces of the Hilbert space, were introduced in the work of Ruskai \cite{ruskaiExchange,ruskai-polatsek} with the motivation of constructing quantum codes that protect against particle exchange errors. These errors occur due to the nature of identical particles, and better isolation from the noisy environment does not reduce the probability of decoherence. Due to the symmetric structure of permutation-invariant codes, they can tolerate exchange (permutation) errors.

Subsequently, it was shown that PI codewords form ground states of the ferromagnetic Heisenberg model and thus can be realized in a Heisenberg ferromagnet \cite{ouyangPI}. Another appealing feature of PI codes is their capability of correcting quantum deletion errors \cite{ouyangDeletion,hagiwaraDeletion}, which is not the case for conventional stabilizer codes. More recently, it was discovered that PI codes support
implementation of transversal gates from finite groups beyond the Clifford group \cite{kubischtaUnitaryt,exoticGates,ouyang2025measurementfreecodeswitchinglowoverhead}.

\begin{definition}\label{def: PI Codes} A permutation invariant code $Q^{\mathrm {PI}}$ is an $N$-qudit quantum code stabilized by the symmetric group: for all $\ket{\psi}\in Q^{\mathrm {PI}}$ and $g\in S_N$, 
\(
    g\ket{\psi}=\ket{\psi}.
\)
\end{definition}
PI codes are often described using {Dicke states}. To define them, we recall that the composition of a string $\ux$ of length $N$ over a $q$-ary alphabet is a $q$-tuple
\begin{equation*}
    \mathrm{C}(\ux) =(n_0(\ux),n_1(\ux),\ldots,n_{q-1}(\ux)),
\end{equation*}
where
\begin{equation*}
    n_i(\ux)= \sum_{j=1}^N\delta_{x_j,i}
\end{equation*}
is the number of occurrences of the character $i$ in the string $\ux$. Note that $\sum_x\mathrm{C}(\ux)=N$. A qudit {\em Dicke state} is the linear combination of all qudit states with the same composition, i.e., 
\begin{equation*}
    \ket{D_{\un}} = \frac{1}{\sqrt{\binom{N}{\un}}}\sum_{\mathrm{C}(\ux)=\un}\ket{\ux}.
\end{equation*}
For example, the Dicke state corresponding to the composition string $(1,1,1)$ is 
\begin{align*}
    \ket{D_{(1,1,1)}} = \frac{1}{\sqrt{6}}(\ket{0,1,2}+\ket{0,2,1}+\ket{1,0,2}+\ket{1,2,0}+\ket{2,0,1}+\ket{2,1,0}). 
\end{align*}
By construction, such states are permutation invariant: for any $\un \in \calS_{q,N}$, we have $g\ket{D_{\un}}=\ket{D_{\un}}$  for all $g\in S_N$. Therefore, the code space of a permutation-invariant code is formed of linear combinations of Dicke states with complex coefficients. A $K$-dimensional qudit PI code $\calQ^{\mathrm {PI}}$ of length $N$ is defined as the subspace spanned by  the following basis states:
\begin{equation*}
    \ket{c_i}=\sum_{\un \in \calS_{q,N}}\alpha_{\un}^{(i)}\cdot \ket{D_{\un}}, \quad i=0,\dots,K-1.
\end{equation*}
 The first family of qubit PI codes correcting an arbitrary number of errors was found in \cite{ouyangPI}. It can be defined using integer parameters $g,n,u$, and so is called the \textit{gnu codes}. The shortest gnu code correcting $t$ Pauli errors has length at least $(2t+1)^2$. PI codes encoding multiple qudits were presented in \cite{ouyangHigherDimensions}. These works introduced an explicit construction of $K$-dimensional qudit PI codes of length $(K-1)(2t+1)^2$  correcting $t$ Pauli errors for any alphabet size $q\geq2$. A more efficient construction of qubit PI codes generalizing gnu codes was found in \cite{aydin2023family}. These qubit codes encode two logical qubits, correct $t$ errors, and require $(2t+1)^2-2t$ physical qubits. Recently, a general framework for constructing 
 high-dimensional qudit PI codes using classical $\ell_1$ codes was suggested in \cite{aydin2025quantum}. Relying on it, the authors of \cite{aydin2025quantum} constructed a family of $K$-dimensional qudit PI codes of length $N$, distance $d=t+1$, and alphabet size $q=N$, where $N=(K-1)t(t+1)$. They also 
 constructed a family of PI codes of growing length $N$ with $q=N$, distance close to $N/\log N$, and dimension $2^N$. In the same paper, the authors established a connection between Fock state codes and PI codes. This result can be summarized as follows.
\begin{prop}\label{prop:FocktoPI} {\rm \cite{aydin2025quantum}}
Let $C_N\in\calS_{q,N}$ be a classical $\ell_1$ code with distance $t+1$. Let $Q$ be a $K$-dimensional Fock state code with total excitation $N$ constructed from $C_N$ as in \eqref{eq:CodeBasis}. Define the mapping
\begin{equation*}
    \ket{\un}_b\underset{h^{-1}}{\stackrel{h}{\Longleftrightarrow}}\ket{D_{\un}}.
\end{equation*}
Applying this mapping to a Fock state code $Q$ yields a PI code with alphabet size $q$, length $N$, and distance $d=t+1$ defined by the basis
\begin{align}\label{eq:basisPICode}
     \ket{c_i} = \sum_{\un\in I_i}\alpha_{\un}^{(i)
 }\cdot \ket{D_{\un}}, \quad i \in \{0,1,\ldots, K-1\}.
\end{align}
\end{prop}
\begin{proof}[Proof sketch]
We start with the code in \eqref{eq:CodeBasis}, which is a $q$-mode Fock state code with total excitation $N$. The mapping $h$ sends a $q$-mode Fock state with total excitation $N$ to a Dicke state with alphabet size $q$ and total length $N$.  Since the Fock state code $Q$ in \eqref{eq:CodeBasis} is constructed using the framework in \cite{aydin2025quantum}, the coefficients $\{\alpha_{\un}^{(i)}\}$ associated to the partition $I_0,I_1,\ldots,I_{K-1}$ of $C_N$ must be a solution for the set of equations (C1)-(C4) of that paper. As shown in Theorem 1 on p.7 of \cite{aydin2025quantum}, these equations guarantee that the PI code in \eqref{eq:basisPICode} has distance $d=t+1$.
\end{proof}

Proposition \ref{prop:FocktoPI} combined with Theorem \ref{th:ExactAG} gives a construction of asymptotically good PI codes.
\begin{tcolorbox}[width=1\linewidth, left=1mm, right=1mm, boxsep=0pt, boxrule=0.5pt, sharp corners=all, colback=white!95!black]
 \begin{theorem}[Asymptotically good PI codes]\label{th:ExactPI} For any $0<\delta<1$, $\varepsilon>0$, and $\alpha>0$ there exists a sequence of PI codes  $(Q^{\mathrm {PI}}_N)_N$, $Q^{\mathrm {PI}}_N\subseteq\calH_{\alpha N, N}$ of increasing length $N$, distance at least $\floor{\delta N}$, and dimension ${K_N}$  as long as 
 \[\frac{1}{N}\log_2 {K_N}\leq  R_{\sr{GV}}(\delta,\alpha)-(\alpha+\delta) h_2\Big(\frac{\alpha}{\alpha+\delta}\Big)+o(1),\]
 for some fixed $o(1)$ function.
\end{theorem}
\end{tcolorbox}
We note that by the previous results, the codewords of the codes $(Q^{\mathrm {PI}}_N)_N$ are balanced in the sense that every symbol of the alphabet is used at most $O(\log N)$ times.

\subsection{Spin codes}
Both PI codes and multimode Fock state codes encode quantum information into a composite many-body quantum system. Another physical platform to realize quantum codes is nuclear state spaces of molecules, atoms, or ions. Spin codes, introduced in \cite{gross}, aim to protect quantum information by encoding it into a single large spin.  Spin-$J$ systems are generally described using $(2J+1)$-dimensional irreducible representations of $SU(2)$. Therefore, spin codes were initially defined and studied using $SU(2)$ geometry. The physical Hilbert space of these platforms has the basis $\{\ket{J,m}\, \quad m=-J,-J+1,\ldots, J-1, J \}$, where each basis state $\ket{J,m}$ is an eigenstate of the angular momentum operator $J_z$ with eigenvalue $m$. It is also possible to interpret the same physical Hilbert space using irreducible representations of $SU(q)$ of the same dimension. The basis states of this Hilbert space can also be defined using vectors in
the discrete simplex $\calS_{q,N}$. In other words, the $D$-dimensional irreducible representation of $SU(q)$ is defined using the basis $\{\ket{\un}_s\!,\, \un\in \calS_{q,N}\}$, where $D=\binom{q+N-1}{q-1}$. Hence, a spin-$J$ system that houses an irreducible representation of $SU(2)$ also houses an irreducible representation of $SU(q)$ for any $q,N>0$, where $2J+1=\binom{q+N-1}{q-1}$. 

A $K$-dimensional $SU(q)$ spin code with total spin $N$ has the following basis states:
\begin{align}\label{eq:basisSpinCode}
    \ket{c_i} = \sum_{\un\in\calS_{q,N}}\alpha_{\un}^{(i)}\ket{\un}_s, \quad i\in[K]
\end{align}
The noise process for spin codes is described by small random rotations. Let  $\calJ_q = \{J_i : i=1,2,\ldots q^2-1\}$ be the generators of the Lie algebra $\mathfrak{s u}(q)$ in the fundamental irreducible representation. When $q=2$, the generators $\calJ_2=\{J_x,J_y,J_z\}$ are the angular momentum operators. Then we say that a spin code has a distance $d=t+1$, if it detects errors from the set
\begin{align*}
    \calE_{t}^{(s)} = \{J_{i_1}J_{i_2}\ldots J_{i_t} : J_{i_k}\in \calJ_q\}
\end{align*}
Relying on the results of \cite{aydin2025quantum}, it is possible to show that Proposition \ref{prop:FocktoPI} can be modified to yield 
a connection between spin codes and Fock state codes. In conclusion, the construction of asymptotically good Fock state codes yields a family of spin codes with increasing $J$ that detect errors from the set $\calE_t^{s}$ with $t$ growing linearly with $N$.

\section{Classical \texorpdfstring{$\ell_1$}{} codes}\label{sec:ClasicalELl_1}
In this section, we state and prove results about the asymptotic behavior of classical $\ell_1$ codes used in our main constructions.

\subsection{Deterministic \texorpdfstring{$\ell_1$}{} codes}\label{sec:GV} GV bounds for codes in $\calS_{q,N}$ were previously considered in the literature. 
A standard GV argument implies that, in a finite metric space $X$, there is a code of size $\ge |X|/{\bar B(d-1)}$, where the denominator is the {\em average volume} of the ball of radius $d-1$ in $X$ \cite{gu1993generalized,tolhuizen1997generalized}. 
Coupled with a recent work \cite{goyal2024gilbert}, which provides the asymptotic growth rate of $\bar B(\cdot)\subset \calS_{q,N}$, this gives an existence bound for sequences of $\ell_1$ codes. However, this result lacks certain features that we need for the construction. Specifically, we are interested in codes
whose codewords have bounded support or the $\ell_\infty$ norm, neither of which is implied by the above argument.
At the same time, choosing a point in $\calS_{q,N}$ at random will yield such points with high probability, so we can limit our sample space to a subset
of ``typical'' vectors and apply greedy search or randomized constructions only on this subset. 
When used in the construction of quantum codes that we employ below, this results in quantum codewords with bounded excitation per mode, yielding a well-balanced Fock state code. 

Limiting ourselves to the subset of typical vectors, we consider approaches to constructing $\ell_1$ codes. The first one relies on a Gilbert greedy
procedure on the subset of typical vectors, resulting in a nonconstructive existence proof of asymptotically good codes. We also propose a simple randomized construction that yields similar code families relying on a low-complexity procedure.

Given a classical code $C\subset \calS_{q,N},$ we define its rate as 
   \begin{equation}\label{eq: rate}
   R(C)=\frac{\log_2 |C|}{\log_2|S_{q,N}|}.
   \end{equation}
For a sequence of codes $(C_N)_N$, $N\to\infty$, the asymptotic rate and relative distance are defined as $\liminf_{N\to\infty} R(C_N)$ and $\liminf_{N\to\infty}d(C_N)/N$, respectively. Note that the diameter of the metric space $S_{q,N}$ is $N$, so the relative distance is at most 1.

For $\ux\in \calS_{q,N}$ let \(
    \sr{B}(\ux,r):= \ppp{\uy\in \calS_{q,N} ~:~ d(\ux,\uy)\leq r}
\)  
   denote the metric ball centered at $\ux$. 
Below we use typical properties of vectors in $\calS_{q,N}$ with respect to the support, $\supp(\ux):=\ppp{i\in [q] : x_i\neq 0}$ and $l_\infty$ norm, $\norm{\ux}_\infty:=\max_{i\in [q]} |x_i|$, collected in Lemmas~\ref{lem:Tailmax} to \ref{lem:BallSizeBound} in the Appendix.
We show that with high probability, the support of a random vector is about $q/(1+\alpha)$, and, perhaps more unexpectedly, its coordinates are bounded above
by $O(\log N)$.

\begin{tcolorbox}[width=1\linewidth, left=1mm, right=1mm, boxsep=0pt, boxrule=0.5pt, sharp corners=all, colback=white!95!black]
\begin{prop}[GV-construction of balanced $\ell_1$ codes]\label{prop:GVbound}
     Given $\alpha>0,\varepsilon>0$ and $0<\delta<1$, consider the function 
     \[ 
     R_{\mathrm{GV}}(\delta)=(1+\alpha)h_2\Big(\frac{1}{1+\alpha}\Big)
     -{\Big(\frac{\alpha}{1+\alpha}+\delta\Big)h_2\Big(\frac{\alpha}{\alpha+(1+\alpha)\delta}\Big)}-(\alpha+\delta)h_2\Big(\frac{\alpha}{\alpha+\delta}\Big).
     \]
     There exists a sequence of $\ell_1$ codes $(C_N)_N$, where $C_N\subseteq \calS_{\alpha N,N}$, such that 
     \begin{enumerate}
         \item $\log_2|C_N|\geq N(R_{\mathrm{GV}}(\delta) +o(1))$;
         \item $C_N$ has distance at least $\floor{\delta N}$;
         \item $\max_{\un\in C_N}|\frac{1}{N}\supp(\un)-\frac{\alpha}{1+\alpha}|= o(1)$; 
         \item $\max_{\un\in C_N}\norm{\un}_\infty<(1+\varepsilon)\log_{1+\alpha} N$. 
     \end{enumerate}
\end{prop}
\end{tcolorbox}

\begin{proof}
By abuse of notation, below we assume that $\alpha N$ is an integer.  Let us fix some ${0<\xi<1}$ 
and let  $\calS_{\alpha N,N}'$ be the set of all $\un\in\calS_{\alpha N,N}$ that satisfy 
    \begin{equation}\label{eq: typical}
    \Big|\frac{1}{N}|\supp(\un)|-\frac{\alpha}{1+\alpha}\Big|\leq N^{-\frac{1}{2}+\xi}, \quad \norm{\un}_\infty \leq (1+\varepsilon)\log_{1+\alpha} N. 
    \end{equation}
{\bf1}. First, we observe that by limiting ourselves to $\calS_{\alpha N,N}'$ we remove only a negligible part of the space, in other words that $\frac{|\calS_{\alpha N,N}'|}{|\calS_{\alpha N,N}|}=1-o(1)$ as $N$ increases. 
   Indeed, let $\su{X}$ be uniform on $\calS_{\alpha N,N}$ and consider the event $\mathcal{E}$ defined in \eqref{eq: typical} with $\su X$ replacing $\un$.
   We will show that $\P(\cE)=1-o(1)$. 

Let us start with $\supp(\su{X})$.
By Lemma~\ref{lem:inftyNormBound}, for sufficiently large $N$, for some decaying sequence $\gamma_N\to 0$, we have
    \begin{align*}
       \P\Big[ \Big|\frac{1}{N}|\supp(\su{X})| & - \frac{\alpha}{1+\alpha}\Big| > N^{-\frac12+\xi/2}\Big]\\
       &\leq \P\pp{\abs{\frac{1}{\alpha N}|\supp(\su{X})|- \E\pp{\frac{1}{\alpha N}|\supp(\su{X})|}}>N^{-\frac{1}{2}+\xi/2}\cdot(1-\gamma_N)}
        \\&\leq \frac{N(N-1)(\alpha N-1) N^{1-\xi}}{ \alpha N(N(1+\alpha)-1)^2(N(1+\alpha)-2)(1-\gamma_N)^2}=O(N^{-\xi}). 
    \end{align*}

 Now let us address the second part of the event $\mathcal E$.       By Lemma~\ref{lem:Tailmax}, for any integer $B>0$ 
    \begin{align*}
        \P\pp{\norm{\su{X}}_\infty \geq B}&\leq \alpha N \frac{\binom{N(1+\alpha)-B-1}{\alpha N-1}}{\binom{N(1+\alpha)-1}{\alpha N-1}}=\alpha N\cdot\prod_{i=0}^{B-1}\frac{N-i}{N(1+\alpha)-i-1} \\
        &\leq \alpha N \p{\frac{N}{N(1+\alpha)-B}}^B.
    \end{align*}
Let us take $B=(1+\varepsilon)\log_{1+\alpha}N$ and let $\gamma_N=B/N.$ For a sufficiently large $N$, we have
    \begin{align*}
        \frac{N}{N(1+\alpha)-B}&= \frac{1}{1+\alpha+\gamma_N} < (1+\alpha)^{-(1-\frac\varepsilon 2)},
    \end{align*}
        and therefore
    \begin{align}
        \P\pp{\norm{\su{X}}_\infty \geq B}&\leq \alpha N(1+\alpha)^{-(1-\frac\varepsilon 2)B}=\alpha N^{(\epsilon^2-\epsilon)/2} \stackrel{N\to\infty}
        \longrightarrow 0.
  \label{eq:InftyNormBoundProb}
    \end{align}
We have shown that both parts of $\mathcal E$ occur with probability $1-o(1),$ but then so does their intersection, the event $\mathcal E$ itself.

{\bf2}. Let us turn to code construction. We start with an arbitrary $\un_1\in \calS_{\alpha N,N}'$ and then choose an arbitrary vector $\un_2$ from $\calS_{\alpha N,N}'\setminus \sr{B}(\un_1,\floor{\delta N})$. Arguing inductively, in the $i$th step we pick an arbitrary codeword $\un_i\in \calS_{\alpha N,N}'\setminus \bigcup_{j=1}^{i-1}\sr{B}(\un_i,\floor{\delta N})$ as long as this set is not empty. Thus, to complete the proof, it is sufficient to prove that,
as $N$ increases,
    \[
    \frac{|\calS_{\alpha N,N}'|}{\max_{\un\in \calS_{\alpha N,N}'}|\sr{B}(\un,\floor{\delta N})|}\geq 2^{N\p{R_{\mathrm{GV}}(\delta)+o(1)}}.
    \]
By Part {\bf 1} above and \eqref{eq:asymbinom},
    \begin{align}\label{eq: S'}
    |\calS_{\alpha N,N}'| &=|\calS_{\alpha N,N}|(1+o(1))=\binom{N(1+\alpha)-1}{\alpha N-1}(1+o(1)) \notag\\
    &=2^{N(1+\alpha)\p{h_2\p{\frac{1}{1+\alpha}}+o(1)}}.
    \end{align}
Next, let us fix $\un\in\calS'_{\alpha N,N}$ and denote $s_\un:=|\supp(\un)|\le \frac {\alpha N}{1+\alpha}(1+\gamma_N)$, where $\gamma_N\downarrow 0$ (not necessarily the same as above), where the support size is given in \eqref{eq: typical}. Then, using Lemma \ref{lem:BallSizeBound},
    \begin{align*}
        \abs{\sr{B}(\un,\floor{\delta N})}&\leq \sum_{j=0}^{\floor{\delta N}}\binom{j+s_{\un}-1}{s_{\un}-1}  \binom{j+\alpha N-1}{\alpha N-1}\\
        &\leq \delta N \max_{0\le j\le \floor{\delta N}}\binom{j+s_{\un}-1}{s_{\un}-1}  \binom{j+\alpha N-1}{\alpha N-1}\\
        &=\delta N  \binom{\floor{\delta N}+s_{\un}-1}{s_{\un}-1}  \binom{\floor{\delta N}+\alpha N-1}{\alpha N-1}.
    \end{align*}
Letting $N\to \infty$ and again using \eqref{eq:asymbinom}, we obtain
   \begin{align*}
       \frac 1N\log_2\abs{\sr{B}(\un,\floor{\delta N})} 
       \le
        \Big(\frac{\alpha}{1+\alpha}+\delta\Big)h_2\Big(\frac{\alpha}{\alpha+(1+\alpha)\delta}\Big)+
        (\alpha+\delta)h_2\Big(\frac{\alpha}{\alpha+\delta}\Big) +o(1)
    \end{align*}
Since this is true for any $\un$, the proof follows by combining this inequality and \eqref{eq: S'}.
\end{proof}

\subsection{Random \texorpdfstring{$\ell_1$}{} codes}\label{sec:randEll1}
In this section, we consider randomized procedures of constructing $\ell_1$ codes.  Although still asymptotically good, their rate–distance trade-offs are inferior to those obtained in the previous section. At the same time, these random ensembles offer a distinct structural advantage. As we will show in the next section, their inherent randomness enables a simple and low-complexity procedure for constructing asymptotically good bosonic Fock state codes, bypassing
a computationally involved step required in an earlier construction, \cite{aydin2025quantum}. See Remark~\ref{remark: Tverberg} for additional details.

We introduce two randomized $\ell_1$ code constructions. The first one samples codewords independently and uniformly from the simplex $\calS_{q,N}$. Using arguments similar in spirit to the GV-type analysis, we show that with high probability the resulting code achieves linear distance and positive rate while satisfying an $O(\log N)$ bound on the coordinate values. We then show that replacing the uniform simplex distribution with a multinomial distribution allows us to further reduce the maximal coordinate value to $O(\log N/\log\log N)$, at the cost of a weaker achievable rate. We emphasize that this flexibility in choosing the per-codeword sampling distribution is quite general, and other distributions can be used within the same framework. In particular, identifying a distribution that simultaneously achieves linear distance and rate while ensuring $O(1)$ boundedness remains an intriguing and important open question.

\begin{prop}[Uniform $\ell_1$ codes] \label{prop:randEll_1} Let
$R_{\mathrm{U}}:[0,1] \to \R$ be a function given by
     \[
       R_{\mathrm{U}}(\delta)=\frac{(1+\alpha)}{2}h_2\p{\frac{1}{1+\alpha}}
     -(\alpha+\delta)h_2\p{\frac{\alpha}{\alpha+\delta}}. 
     \]
     Let $C_N^{\sr{U}}=\set{\su{n}^{(i)}}_{i=1}^{L_N}$ be a code formed of $L_N$ uniform random vectors chosen independently from the simplex $\calS_{\alpha N,N}$. For $N\to\infty,$ if
    \begin{align}\label{eq:RandCond}
        \log_2{L_N}\leq N\p{R_{\mathrm{U}}(\delta) + o(1)},
    \end{align} 
    then $\lim_{N\to \infty }\P\pp{d(C_N^{\sr{U}})\geq \floor{\delta N}}=1$. 
\end{prop}
\begin{proof}
Let $t=\floor{\delta N}$ and for $i,j\in [L], i\neq j$, consider a random variable $\s{Y}_{ij}=\Ind(d\p{\su{n}^{(i)},\su{n}^{(j)}} \leq t).$
We have
    \begin{align}\label{eq:Yij}
        \P\Big[\sum_{i\neq j}\s{Y}_{i,j}=0\Big]=1-\P\Big[\bigcup_{i\neq j} \set{\s{Y}_{i,j}=1} \Big]\geq 1-\sum_{i\neq j}\P\pp{\s{Y}_{ij}=1}.
    \end{align}
 Using independence, we compute
    \begin{align*}
        \P\pp{\s{Y}_{ij}=1}&=\sum_{\un \in \calS_{\alpha N,N}}\P\pp{\s{Y}_{ij}=1 ~|~\su{n}^{(i)}=\un}\P\pp{\su{n}^{(i)}=\un}    \\
        & =\sum_{\un}\P\pp{\un^{(j)}\in \sr{B}(\un,t) ~|~\su{n}^{(i)}=\un}\P\pp{\su{n}^{(i)}=\un}\\
        & =\sum_{\un}\P\pp{\un^{(j)}\in \sr{B}(\un,t)}\P\pp{\su{n}^{(i)}=\un}\\
        &= \sum_{\un} \frac{\abs{\sr{B}(\un,t)}}{\abs{\calS_{\alpha N,N}}}  \P\pp{\su{n}^{(i)}=\un}\\
        &{\leq} \sum_{\un \in \calS_{\alpha N,N}}  t  \frac{\binom{\floor{\delta N}+\alpha N-1}{\alpha N-1}^2}{\binom{N(1+\alpha)-1}{\alpha N-1}}\cdot \P\pp{\su{n}^{(i)}=\un}\\
     &=2^{-N(2R_{\sr{U}}(\delta)+o(1))}
        \end{align*}
where the inequality follows from Lemma~\ref{lem:BallSizeBound}, and the last line uses \eqref{eq:asymbinom}.
 Combining this estimate with \eqref{eq:Yij} and the assumption in \eqref{eq:RandCond}, we find that
        \begin{align*}
            \P\Big[\sum_{i\neq j}\s{Y}_{i,j}=0\Big]&\geq 1-\sum_{i\neq j}\P\pp{\s{Y}_{i,j}=1}\geq 1-\binom{L}{2}\cdot 2^{-N(2R_{\sr{U}}(\delta)+o(1))}\\
            &\geq 1-  2^{2\p{\log_2L-N\p{R_{\su{U}}(\delta)+o(1)}}}=1-o(1).
        \end{align*}
We have shown that the distance of the code with probability approaching one is at least $\floor{\delta N}$.
\end{proof}

\begin{remark} Note that uniform sampling from $S_{q,N}$ is easily done by the balls-and-bins argument. To start, sample numbers $b_1,b_2,\dots,b_{q-1}$ uniformly without replacement from the set $\{1,\dots,N+q-1\}$. W.l.o.g. assume that $b_1<b_2<\dots<b_{q-1}$. Then the vector $(x_1,\dots,x_q)$, where $x_1=b_1-1$, 
$x_i=b_i-b_{i-1}-1, i=2, \dots,q-1$, and $x_q=(N+q-1)-b_{q-1}$, is a uniform random vector on $\calS_{q,N}$. 
\end{remark}

To further improve the guaranteed per-mode occupancy, we replace the uniform simplex distribution with a multinomial distribution. Specifically, each codeword is generated by independently placing $N$ balls into $q$ bins, where each ball is assigned to a bin uniformly at random, and the resulting occupancy vector is taken as the codeword.

\begin{construction}\label{construction: uniform}
Let $\s{Y}_1,\dots,\s{Y}_N$ be i.i.d. random variables, distributed uniformly on the set $[q]$. Define a random vector $\su{X}\in\calS_{q,N}$ by 
 \begin{align}
     \s{X}_i=\sum_{j=1}^N \Ind_{\ppp{\s{Y}_j=i}},  \quad i=0,\dots,q-1   \label{eq:MultinomialDist}
 \end{align}
We refer to its distribution as $\s{Multinomial}(q,N)$. A code in $\calS_{q,N}$ of size $L$ is constructed by sampling $L$ independent vectors $\su{X}_i$ according to this procedure.
\end{construction}

This construction relies on independent sampling and is therefore simpler to implement than uniform sampling from the simplex, while still producing codewords with well-behaved typical properties. We start by showing that the resulting codes with high
probability are asymptotically good. At the end of this section (Lemma~\ref{lemma: sublog}) we also prove that the coordinates of a random vector
$\su{X}\sim \s{Multinomial}(q,N)$ are uniformly bounded above by a sublogarithmic function of $N$.

\begin{prop}[Multinomial $\ell_1$ codes] \label{prop:randEll_1Multi}
Fix $\delta\in(0,1)$ and $\alpha>0$, and define the function $ R_{\mathrm{M}}:[0,1]\to \R$ by
     \[ R_{\mathrm{M}}(\delta)=\frac{(\Delta_\alpha-\delta)^2}{8\ln 2}, \] 
     where 
     \[
     \Delta_\alpha:=\frac{1}{\pi}\intop_0^{\pi}  e^{-\frac{2(1-\cos(\theta))}{\alpha}}(1-\cos(\theta))d\theta.
     \]
     Let $C_N^{\sr{M}}=\set{\su{n}^{(i)}}_{i=1}^{L_N}$ be a code formed of $L_N$ i.i.d. random vectors sampled out of $\calS_{\alpha N,N}$ according to $\s{Multinomial}(\floor{\alpha N},N)$.  There exists an $o(1)$ function such that for $N\to\infty,$ if,
    \begin{align}\label{eq:RandCond2}
        \log_2L_N\leq N\p{R_{\mathrm{M}}(\delta) + o(1)}
    \end{align} then $\lim_{N\to \infty }\P\pp{d(C_N^{\sr{M}})\geq \floor{\delta N}}=1$. 
\end{prop}

For the proof of Proposition~\ref{prop:randEll_1Multi}, we will need the expected distance between two random vectors as above.
It is given in the following lemma.
\begin{lemma}\label{lem:Bessel}
    Let $\su{X}$ and $\su{Z}$ be independent $\s{Multinomial}(\floor{\alpha N},N)$ random vectors. Then, 
    \begin{align}\label{eq:BesselExp}
        \nonumber\lim_{N\to \infty}\frac{1}{N}\E\pp{d(\su{X},\su{Z})}&=\Delta_\alpha.
    \end{align}
\end{lemma}
\begin{proof} Although $\su{X}$ and $\su{Z}$ depends on $N$, below we will suppress $N$ from the notation. As before, by abuse of notation, we assume that $\alpha N$ is an integer. Observe that 
    \begin{equation}\label{eq:expCalc1}
        \frac{1}{N}\E\pp{d(\su{X},\su{Z})}=
        \frac{1}{2N}\sum_{i=0}^{{\alpha N}-1}\E{\abs{\s{X}_i-\s{Z}_i}}=\frac{{\floor{\alpha N}}}{2N}\cdot \E\abs{\s{X}_1-\s{Z}_1},
    \end{equation}
   where the last equality follows since the pairs $(\s{X}_i,\s{Z}_i)$ are identically distributed and therefore all have the same expectation. From the definition of the multinomial distribution, $\s{X_i}$ and $\s{Z}_i$ are independent $\sr{Binomial}(N,p_N)$, with $p_N=\frac{1}{{\alpha N}}$. As is well known (\cite[Sec.~VI.5]{feller1968introduction}), such binomial random variables converge in distribution to a $\sr{Poi}(\frac{1}{\alpha})$ random variable, and thus their distribution functions  $F_{\s{X}_1},F_{\s{Z}_1}$ converge to $F_{\alpha}$, the cumulative distribution function of $\sr{Poi}(\frac{1}{\alpha})$. By the independence of $\s{X}_1$ and $\s{Z}_1$, the joint distribution satisfies 
    \[F_{\s{X}_1,\s{Z}_1}(x,z)=F_{\s{X}_1}(x)F_{\s{Z}_1}(z)\to F_{\alpha}(x)F_{\alpha}(z), 
    \]
    and therefore as $N\to\infty$, the pairs $(\s{X}_1,\s{Z}_1)$ converge in distribution to $(\s{X},\s{Z})$, distributed as a pair of independent Poisson random variables $\sr{Poi}(\frac{1}{\alpha})$. By the continuous mapping theorem (see Lemma~\ref{lem:ctsMap}), we now can claim that as $N$ increases, the random variables $|\s{X}_1-\s{Z}_1|$ converge in distribution to $|\s{X}-\s{Z}|$. Next, note that this sequence of random variables is uniformly integrable. Indeed, since $\s{X}_1$ and $\s{Z}_1$ are non-negative, we have:
    \begin{align*}
        \E|\s{X}_1-\s{Z}_1|\leq \E\pp{\s{X}_1+\s{Z}_1}=2 \frac{N}{\floor{\alpha N}}<\frac{2}{\alpha}+\varepsilon,
    \end{align*}
    for sufficiently large $N$, as desired. By Lemma~\ref{lem:ConvExp}, this suffices to claim that $\E|\s{X}_1-\s{Z}_1|\to \E|\s{X}-\s{Z}|$ as $N\to\infty$.

    It now remains to compute the expectation of $|\s{X}-\s{Z}|$, which is the absolute value of $\s{S}\triangleq
    \s{X}-\s{Z}$-the difference between two independent $\sr{Poi}(\frac{1}{\alpha})$. The distribution of $\s{S}$ is known as $\sr{Skellam}(\frac{1}{\alpha},\frac{1}{\alpha})$ distribution \cite{skellam1946frequency}, and its PDF is given by 
    \[P_{\s{S}}(k)=e^{-\frac{2}{\alpha}}I_k\p{\frac{2}{\alpha}}, \quad k\in \Z,\]
  where for $k\in \Z$, $I_k$ is the $k^{\rm th}$ order modified Bessel function of the first kind, 
    \[I_k(x):=\frac{1}{\pi}\intop_{0}^{\pi} e^{x \cos\theta} \cos(k\theta)d\theta.\]  
    Note that the above distribution is symmetric around zero, and therefore
    \begin{align*}
        \E[\abs{\s{S}}]&=\sum_{k\in \Z}|k|\cdot P_{\s{S}}(k)=2\cdot \sum_{k=1}^{\infty}k\cdot P_{\s{S}}(k)=e^{-\frac{2}{\alpha}}\sum_{k=1}^{\infty}k I_k\p{\frac{2}{\alpha}}\\
         &\overset{(a)}{=}\frac{e^{-\frac{2}{\alpha}}}{\alpha}\lim_{m\to \infty }\sum_{k=1}^{m} \p{I_{k-1}\p{\frac{2}{\alpha}}-I_{k+1}\p{\frac{2}{\alpha}}}\\
        &=\frac{e^{-\frac{2}{\alpha}}}{\alpha}\lim_{m\to \infty }\sum_{k=1}^{m}\p{I_{0}\p{\frac{2}{\alpha}}+I_{1}\p{\frac{2}{\alpha}}-I_{m}\p{\frac{2}{\alpha}}-I_{m+1}\p{\frac{2}{\alpha}}}\\
         &\overset{(b)}{=} \frac{2e^{-\frac{2}{\alpha}}}{\alpha} \Big(I_{0}\Big(\frac{2}{\alpha}\Big)
         +I_{1}\Big(\frac{2}{\alpha}\Big)\Big),
    \end{align*}
    where $(a)$ uses the following recurrence relation for the Bessel functions \cite[Sec.~10.29]{olver2010nist}: 
    \[\frac{2k}{x}I_{k}(x)=I_{k-1}(x)-I_{k+1}(x), \quad k\in \Z\]
    and $(b)$ follows since $I_{m}(1/\alpha)\xrightarrow[m\to\infty]{}0$ as $P_{\s{S}}(K)$ sums to $1$. Combining the above with \eqref{eq:expCalc1} we conclude:
    \begin{align*}
        \lim_{N\to\infty }\frac{1}{N}\E\pp{d(\su{X},\su{Z})}&=\lim_{N\to\infty }\frac{\floor{\alpha N}}{2N}\cdot \E\abs{\s{X}_1-\s{Z}_1}=\frac{\alpha}{2} \cdot \lim_{N\to\infty } \E\abs{\s{X}_1-\s{Z}_1}\\
        &=\frac{\alpha}{2} \cdot \E\pp{|\s{X}-\s{Z}|}= e^{-\frac{2}{\alpha}}\p{I_{0}
        \Big(\frac{2}{\alpha}}+I_{1}\p{\frac{2}{\alpha}}\Big). \qedhere
    \end{align*}
\end{proof}

We can now prove Proposition~\ref{prop:randEll_1Multi}:
\begin{proof}[Proof of Proposition~\ref{prop:randEll_1Multi}] The proof uses the notation and closely follows the logic of the proof of Proposition~\ref{prop:randEll_1} with the only change in the evaluation of $\P[\s{Y}_{i,j}=1]=\P[d(\su{n}^{(i)},\su{n}^{(j)})<t]$, which we now bound using McDiarmid's inequality. Indeed, let $\ppp{\s{Z}_{i,l} ~:~ i=1,\dots,N, l=1,\dots, L}$ be the $\s{Unif}([q])$ i.i.d. random variables such that $\su{n}^{(i)}$ is defined by $(\s{Z}_{1,i},\dots, \s{Z}_{N,i})\triangleq \su{Z}_i$ via \eqref{eq:MultinomialDist}. Consider the function $f:[q]^{2N}\to \R_+$ defined by 
\[f\p{\s{Z}_{1,i},\dots, \s{Z}_{N,i},\s{Z}_{1,j},\dots, \s{Z}_{N,j}}=\frac{1}{N}d\p{\su{n}^{(i)},\su{n}^{(j)}}.\]
We note that $f$ satisfies the bounded differences property \eqref{eq: bounded differences} with $c_1=c_2=\cdots = c_{2N}=\frac{2}{N}$. Indeed, changing the bin of one ball will increase one coordinate of the corresponding simplex vector by $1$ and decrease another coordinate by $1$. Setting $t=\floor{\delta N}$ and $q=\floor{\alpha N}$, we have that 
\begin{align*}
    \P[ d(\su{n}^{(i)},\su{n}^{(j)})]<\floor{\delta N}]&\leq \P\Big[|d(\su{n}^{(i)},\su{n}^{(j)})-\E \,d(\su{n}^{(i)},\su{n}^{(j)})|>\E\,d(\su{n}^{(i)},\su{n}^{(j)})-\floor{\delta N}\Big]\\[.1in]
    &\overset{(a)}{=}\P\Big[\Big| \frac{1}{N}d(\su{n}^{(i)},\su{n}^{(j)})-\E\Big[\frac{1}{N}d(\su{n}^{(i)},\su{n}^{(j)})\Big]\Big|>R_{\alpha}-\delta +o(1)\Big]\\[.1in]
    &=\P\pp{\abs{f(\su{Z}_i,\su{Z}_j)-\E\pp{f(\su{Z}_i,\su{Z}_j)}}>R_{\alpha}-\delta +o(1)}\\[.1in]
    &\overset{(b)}{\leq} 2\exp\Big(-\frac{2(\Delta_\alpha -\delta +o(1))^2}{\sum_{i=1}^{2N}\frac{4}{N^2}}\Big)\\
    &=\exp\p{-\frac{N}{4}\p{(\Delta_\alpha-\delta)^2+o(1)}},
    \end{align*}
    where $(a)$ follows since $\frac{1}{N}\E\,d(\su{n}^{(i)},\su{n}^{(j)})=\Delta_\alpha+o(1) $ by Lemma~\ref{lem:Bessel}, and $(b)$ follows from  MacDiarmind's inequality (see Lemma~\ref{lem:McDiarmind}). We now conclude following the steps of Proposition~\ref{prop:randEll_1}:
\begin{align*}
    \P\pp{d(\s{C}_N^{\sr{M}})\geq \floor{\delta N}}&\geq 1-\sum_{i\neq j}\P\pp{\s{Y}_{i,j}=1}\leq 1-\binom{L}{2}\cdot \exp\p{-\frac{N}{4}\p{(R_\alpha-\delta)^2+o(1)}}\\
    &>1-2^{2 \log_2 L-N \frac{(R_{\alpha}-\delta)^2+o(1)}{4\ln 2}}=1-o(1),
\end{align*}
where the last equality from the assumption \eqref{eq:RandCond2}. This completes the proof.
\end{proof}

 In Figure~\ref{fig: random bounds}, we plot the curves $R_{\mathrm  GV}(\delta)$ and $R_{\mathrm  U}(\delta)$ for $\alpha=1$. The bound $R_{\mathrm  M}(\delta)$ is positive for $\delta\le 0.292$; however, the values of the rate are much smaller than the other two bounds, so we are not showing it in the same plot (the loss occurs because we rely on a concentration inequality). We also remark that the {\em code rate} as defined in \eqref{eq: rate} above is obtained by dividing the values $R_{\mathrm GV}$ etc. by $\frac1N\log_2 |S_{\alpha N, N}|\sim (1+\alpha) h_2(1/(1+\alpha))$, which does not change the conclusion regarding asymptotic goodness of the code sequences. We also observe that \cite[Thm.~10]{goyal2024gilbert} gives a better asymptotic bound for codes in the entire simplex, not accounting for the typical subset restriction.

\begin{figure}[ht]
\centering\includegraphics[width=0.4\linewidth]{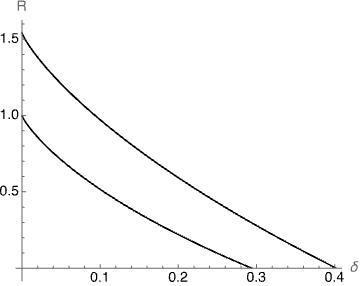}
\caption{Bounds $R_{\mathrm  GV}(\delta)$ (the upper curve) and $R_{\mathrm U}(\delta)$ for $\alpha=1$. }\label{fig: random bounds}
\end{figure}

We conclude this section by showing that vectors drawn from the multinomial distribution typically have sub-logarithmic $\infty$-norm, mimicking the properties of the uniform distribution on the simplex. In fact, it turns out that the $\infty$-norm of a multinational random vector decays as $(\ln N)/\ln\ln N$, which is even faster 
than the decay rate of $\ln N$ for a uniform vector.
\begin{lemma}\label{lemma: sublog}
Let $\su{X}\sim \s{Multinomial}(\floor{\alpha N},N)$ be a random vector on the simplex. Then, for any fixed $\varepsilon>0$
\[\P\pp{\norm{\su{X}}_\infty\geq \frac{(1+\varepsilon)\ln N}{\ln\ln N}}=o(1). \]
\end{lemma}
\begin{proof}
    As shown in the proof of Lemma~\ref{lem:Bessel}, each coordinate $\s{X}_i$ is distributed as $\sr{Binomial}(N,p_N)$, with $p_N=\frac{1}{\floor{\alpha N}}$, which converge in distribution to a $\sr{Poi}(\frac{1}{\alpha})$ distribution. Thus, by the union bound
    \begin{align*}
        \P\pp{\norm{\su{X}}_\infty\geq \frac{(1+\varepsilon)\ln N}{\ln\ln N}}&=\P\pp{\bigcup_{i=0}^{\floor{\alpha N}-1}\ppp{\s{X}_i\geq \frac{(1+\varepsilon)\ln N}{\ln\ln N}}}\\
        &\leq \sum_{i=0}^{\floor{\alpha N}-1}\P\pp{\s{X}_i\geq \frac{(1+\varepsilon)\ln N}{\ln\ln N}}\\
        &\leq \alpha N\cdot \p{\P\pp{\sr{Poi}\p{\frac{1}{\alpha}}\geq \frac{(1+\varepsilon)\ln N}{\ln\ln N}}\cdot (1+o(1))},
    \end{align*}
    where the last inequality follows from convergence in distribution. It is therefore remains to show that 
    \begin{align}
        \P\pp{\sr{Poi}\p{\frac{1}{\alpha}}\geq \frac{(1+\varepsilon)\ln N}{\ln\ln N}}=o\p{\frac{1}{N}}\label{eq:poiProbTail}
    \end{align}
    Denote $\lambda=1/\alpha,$ and observe that for any $k\in \N$ by Stirling approximation
    \begin{align*}
     \P\pp{\sr{Poi}\p{\frac{1}{\alpha}}\geq k}   &=\sum_{n=k}^\infty e^{-\lambda }\frac{\lambda^n}{n!}= \lambda\frac{\lambda^k}{k!}\p{1+o\p{\frac{1}{k}}}\\
     &e^{-\lambda }\exp\p{-\lambda+k \ln \lambda-k\ln k+o(k)} \p{1+o\p{\frac{1}{k}}}\\
     &=\exp\p{-k\ln k(1+o(1))}.
    \end{align*}
    Setting $k=(1+\varepsilon)\ln N/\ln\ln N$ we obtain 
    \begin{align*}
        N\cdot \P\pp{\sr{Poi}\p{\frac{1}{\alpha}}\geq \frac{(1+\varepsilon)\ln N}{\ln\ln N}}
        &\leq \exp\p{-(1+\varepsilon)\cdot(\ln N)\p{1-o(1)}+\ln N}\\
        &=N^{-\varepsilon+o(1)}=o(1). \qedhere
    \end{align*}
\end{proof}

\section{Conclusion and outlook}
In this work, we studied quantum error correction for amplitude damping noise from two complementary perspectives. First, we refined the understanding of quantum  
codes' performance when used over the physically relevant, non-truncated AD channel. Second, we constructed new families of codes 
whose performance on this channel can be rigorously analyzed and shown to improve upon the known results. These two aspects, analysis of the physical noise model and explicit code constructions, are closely intertwined and together provide a more complete picture of error correction in this setting.

One conceptual message that emerges from our study concerns the relationship between the abstract framework of quantum coding theory and the performance of codes on concrete physical channels. The Knill–Laflamme conditions provide a powerful and general formalism for reasoning about correctable error sets, and they remain a cornerstone of quantum error-correction theory. At the same time, our results highlight the importance of carefully analyzing how a code behaves under the actual physical noise model of interest. When this connection is not examined directly, one may obtain codes that satisfy formal QECC conditions yet provide limited protection against the underlying physical channel. In this sense, our analysis can be viewed as an amplitude-damping/Fock-state analogue of the well-known approximation principle whereby general independent noise processes are related to restricted error models (see, e.g., \cite[Theorem 1.3.4]{gottesman2024surviving}).

Another observation arising from this work concerns the role of randomness in approximate quantum error correction. Much of the existing literature introduces randomness directly at the level of the \textit{physical Hilbert space}. A canonical example is the framework of Haar-random codes, where the encoding subspace itself is chosen as a random subspace of the physical system \cite{klesseApproximateQuantumError2007,kong2022near,ma2025haar}. A related perspective appears in the construction of codes from many-body eigenstates of translation-invariant spin chains \cite{brandao2019quantum}, where the code space is formed by randomly choosing physical eigenstates of a quantum Hamiltonian within a suitable energy window.

In contrast, the randomness in our construction appears at a different structural level. Rather than randomizing the physical code subspace, we introduce randomness in the \textit{logical indexing geometry} from which the codewords are generated. Concretely, the codewords are indexed by points in an underlying classical metric space (the discrete simplex equipped with the $\ell_1$ metric), and randomness is applied to this geometric structure. The performance guarantees in our analysis arise from concentration phenomena on the metric space. While our construction focuses on the $\ell_1$ geometry of the simplex, the mechanism driving the result is not tied to this specific setting. Similar concentration effects are known to arise in a wide variety of high-dimensional metric spaces, suggesting that introducing randomness at the level of the geometric or combinatorial structure underlying a quantum code may provide a flexible and broadly applicable design principle. Investigating how this perspective can be leveraged in other coding geometries and for different physical noise models remains an intriguing direction for future work.

We conclude with a remark concerning spin codes. The errors in spin codes are typically modeled as small random rotations that transform spin states according to the Lindbladian master equation. We note that in the asymptotic regime, one can address the question similar to the one
treated in this work for the AD channel, namely, relating the continuous quantum channel generated by the Lindbladian to a truncated error set defined by a subset of generators of the Lie algebra $\mathfrak{s u}(q)$. We hope to address this problem in future research.

\begin{acknowledgements}
The research of the first author was supported by a Rothschild Fellowship granted by the Yad HaNadiv Foundation. The research of the last two authors was partially supported by NSF grants CIF-2330909 and CIF-2526035.
\end{acknowledgements}

\bigskip\noindent
{\small{\sf AI usage:} AI tools were used to assist with proofreading the final paper.}

\bigskip\noindent
{\small{\sf Author contributions: All authors contributed to the mathematical development and writing of the manuscript}.

\appendix
\section{Technicalities and auxiliary results}
In this section, we give formal definitions and technical results that we shall use throughout the paper. 
\subsection{Channels and state norms}
\begin{definition}[Trace norm, Diamond norm, Completely Bounded norm]\label{def:norms}
    Let $\calH$ and $\calH'$ be finite-dimensional Hilbert spaces. Let $X:\calH\to \calH'$ and $\calN:L(\calH)\to L(\calH')$ linear maps. 
    \begin{enumerate}
        \item The trace norm on $L(\calH)$ is defined by 
        \[ \norm{X}_1\triangleq \tr(|X|), \quad |X|\triangleq \sqrt{X^\dag X},\]
        and its just the sum of absolute values of singular values of $X$. 
        \item The completely bounded norm of $\calN$ is defined as
        \[\norm{\calN}_{\s{cb}}=\sup_{\substack{X\in L(\calH^{\otimes 2})\\ \norm{X}_1\leq 1}}\norm{\p{I_{L(\calH)}\otimes \calN} (X)}_1. \]
        \item The diamond norm of $\calN$ is defined as 
        \[\norm{\calN}_{\diamond}=\sup_{\rho\in D(\calH^{\otimes 2})}\norm{\p{I_{L(\calH)}\otimes \calN} (\rho)}_1. \]
    \end{enumerate}
\end{definition}
Note that it is immediate from the definition of the diamond and completely bounded norms that $\norm{\calN}_\diamond\leq \norm{\calN}_{\s{cb}}$.
An extension of Fuchs–van de Graaf inequalities (see \cite[Section 6.2]{khatri2020principles}) implies that for two quantum channels $\calN,\calM$ we have
\[ 1-\sqrt{\calF_e(\calN,\calM)}\leq \frac{1}{2}\norm{\calN-\calM}_{\diamond} \leq \sqrt{1-\calF_e(\calN,\calM)}.\]
Translating to Bures distance, the above implies the following: 
\begin{lemma}\label{lem:DiamonaBures}
    For any two quantum channels $\calN,\calM$ we have 
    \begin{align}
     \frac{1}{2}\norm{\calN-\calM}_\diamond \leq d(\calN,\calM)\leq \sqrt{\norm{\calN-\calM}_\diamond}.
     \label{eq:diamond-Bures}
\end{align}
\end{lemma}
We shall also require the Duality principle and a well-known inequality between norms (see \cite[Section 1.1, pp 32]{watrous2018theory}):
\begin{lemma}[The duality principle]\label{lem:Duality}
    For any $X,Y\in L(H_A)$ we have
    \[\tr(Y^\dag X)\leq \norm{X}_1\cdot \norm{ Y}_\infty\]
\end{lemma}
\begin{lemma}\label{lem:HermitianPreserving}
    Let $\calB$ be the channel from Theorem~\ref{th:BenyOreshkov}. There exists $\ket{\psi}\in Q^{\otimes 2}$ such that \[\norm{\calB}_{\diamond}=\norm{\calB}_\s{cb}=\norm{I_{L\p{\calH}}\otimes \calB (\ket{\psi}\bra{\psi})}_1.\] 
\end{lemma}
\begin{proof}
    We begin by noticing a simple fact: if $\calN:L(\calH)\to L(\calH') $ is a Hermitian-preserving map (which means if $\rho$ is Hermitian, then so is $\calN(\rho)$). Then, for some state $\ket{\psi}\in \calH^{\otimes 2}$,
    \[\norm{\calN}_{\diamond}=\norm{\calN}_\s{cb}=\norm{I_{L\p{\calH}}\otimes \calN (\ket{\psi}\bra{\psi})}_1.\]  
    To see that, we recall that Hermitian-preserving operators, the completely bounded norm is maximized by a one-dimensional projection (see \cite[Theorem 3.51]{watrous2018theory}); in other words,
    \begin{align}
        \norm{\calN}_{\s{cb}}=\sup_{\substack{\ket{\psi}\in \calH^{\otimes 2}\\ \braket{\psi|\psi}=1}}\norm{I_{L\p{\calH}}\otimes \calN (\ket{\psi}\bra{\psi})}_1\leq \sup_{\rho\in D(\calH)} \norm{I_{L\p{\calH}}\otimes \calN (\rho)}_1=\norm{\calN}_{\diamond},\label{eq:CBdiamond}
    \end{align}
    and since the opposite inequality is trivially satisfied, equality holds. Note that, as it was shown in \cite{beny2010general}, $\calB$ is the difference between two quantum channels and therefore is Hermitian preserving. Thus, it only remains to show that we can find such a $\ket{\psi}$, satisfying \eqref{eq:CBdiamond}, which also belongs in $C^{\otimes 2}$. Indeed, note that for any $\rho=\sum_{i}\eta_i\otimes \tau_i\in L(\calH^{\otimes 2})$ we can write 
    \begin{align*}
        I_{L(\calH)}\otimes \calB (\rho)&=\sum_{k,l}\sum_i \tr(\tau_i (P A_k^\dag A_l P- \lambda_{k,l}P ))  \cdot \eta_i\otimes \ket{k}\bra{l}\\
        &=\sum_{k,l}\sum_i \tr((P\tau_iP) (P A_k^\dag A_l P- \lambda_{k,l}P )) \cdot \eta_i\otimes \ket{k}\bra{l}\\
        &= (I_{L(\calH)}\otimes \calB)\circ (I_{L(\calH)}\otimes \calP) (\rho),
    \end{align*}
    where $\calP(\tau)= P\tau P $. Note that $(I_{L(\calH)}\otimes \calP)$ is a completely positive trace non-increasing map, whose image is $L(\calH \otimes Q)$. Thus, for any $\rho\in D(\calH^{\otimes 2})$ outside of the kernel, we have that the operator 
    \[ \rho'=\frac{1}{\tr(I_{L(\calH)}\otimes \calP (\rho))}\cdot I_{L(\calH)}\otimes \calP (\rho)\]
    is a density operator in $D(\calH \otimes Q) $, and
    \[\norm{I_{L(\calH)}\otimes  \calB (\rho)}_1= \tr(I_{L(\calH)}\otimes \calP (\rho))\cdot \norm{I_{L(\calH)}\otimes \calB (\rho')}_1\leq \norm{I_{L(\calH)}\otimes \calB (\rho')}_1. \]
    Thus, 
    \begin{align}
        \norm{ \calB}_{\diamond}&=\sup_{\rho\in D(\calH^{\otimes 2})}\norm{I_{L(\calH)}\otimes \calB(\rho)}_1 \notag\\
        &\leq\sup_{\rho'\in D(\calH\otimes Q)}\norm{I_{L(\calH)}\otimes \calB(\rho')}_1 \notag\\
        &=\sup_{\tau \in D( Q^{\otimes 2})}\norm{I_{L(Q)}\otimes \calB(\tau)}_1 \label{eq:diamondC}\\
       &{=} \sup_{\ket{\psi_Q} \in Q}\norm{I_{L(Q)}\otimes \calB(\ket{\psi_Q}\bra{\psi_Q})}_1, \notag
    \end{align}
    where \eqref{eq:diamondC} follows since this expression is exactly the diamond norm of $\calB$ when considered as a channel on $L(Q)$, and the diamond norm is equivalently defined by an additional optimization on all auxiliary reference systems of any finite dimension (see \cite[Theorem 9.1.1]{wilde2013quantum}):
    \[\norm{N}_{\diamond}=\sup_{n\in \N}\sup_{\rho_n\in D(\C^n\otimes \calH)}\norm{I_{L(\C^n)}\otimes \calN(\rho_n)}_1.\]
    Finally, $(b)$ follows from \eqref{eq:CBdiamond}. This completes the proof.

\end{proof}

\subsection{Proof of Theorem~\ref{th:TruncToAD}}
    The idea of the proof is quite simple. Let $\calD_{\leq t}$ be a recovery operation for $Q$ under the truncated noise channel $\calN_{\leq t}$ such that $\calF_e(\calD_{\leq t} \circ\calN_{\leq t}|_Q,I_{L(Q)})\geq 1-\varepsilon^2$, which exists by the assumption. Consider the decoding operation $\calD$ given by the following procedure: 
    \begin{itemize}
        \item Given a state $\calN(\rho)$, perform a measurement to determine whether $\rho$ has lost no more than $t$ photons. 
        \item If that is the case,  use the decoder $\calD_{\leq t}$ on the resulting state. 
        \item Otherwise, decode $\calN(\rho)$ to some arbitrary state $\tau$. 
    \end{itemize}

Let us proceed to implement this idea.
    
{\em Defining the decoder}. Consider a quantum channel 
    \begin{align*}
    \calM: \calH_{q,N}&\to \Span\big(\bigcup_{r\leq N }\calH_{q,r}\big)\otimes \C^2\\
    \rho&\mapsto P_{\leq t} \,\rho\, P_{\leq t}\otimes \ket{\mathrm{0}}\bra{0} +\tr\p{\p{I-P_{\leq t}}\rho}\cdot\tau\otimes \ket{1}\bra{1},  
      \end{align*}
    where $\ket{0}$ and $\ket{1}$ are ancillas representing {\em success} and {\em fail}\ \  flag states, $P_{\leq t}$ is the projection on $\mathrm{Span}\p{\bigcup_{i=N-t}^N \calH_{q,i}} $, and  $\tau=\ket{\psi_\tau}\bra{\psi_\tau}$ is an arbitrary pure state in $Q$. It is easy to check that $\calM$ is completely positive,
    with Kraus operators 
    \[
    E_0=P_{\leq t}\otimes \ket{0}, \quad E_{\un}=\ket{\psi_\tau}\bra{\un}(I-P_{\leq t})\otimes \ket{1}, \quad \un\in \bigcup_{r\leq N}\calS_{q,r}.
    \]
$\calM$ is also trace preserving:
    \begin{align*}
        \tr(\calM(\rho))&= \tr(P_{\leq t} \,\rho\, P_{\leq t}\otimes \ket{\mathrm{0}}\bra{0}) +\tr\p{\p{I-P_{\leq t}}\rho}\cdot\tr(\tau\otimes \ket{1}\bra{1})\\
        &=\tr(P_{\leq t}\rho)+\tr((I-P_{\leq t})\rho)=\tr(\rho).
    \end{align*}
    Let $P_0$ be the projection on the success space $\mathrm{Span}\p{\bigcup_{r\leq N
    }\calH_{q,r}}\otimes \mathrm{Span}(\ket{0})$ and $P_1=I- P_0$. Define $\Tilde{\calD}:\mathrm{Span}\p{\bigcup_{r\leq N
    }\calH_{q,r}}\otimes \C^2 \to \calH_{q,N}$ given by 
    \[\Tilde{\calD}(\rho)= \calD_{\leq t}(\tr_{\C^2}(P_0 \rho P_0))+ \tr_{\C^2}(P_1 \rho P_1),\]
    where $\tr_{\C^2}(\cdot )$ denotes the operation of tracing out the success/failure qubit, given by partial trace (which is CPTP).
    Note that the maps  $\rho \to P_i \rho P_i, i=0,1$ are trivially CP, and thus each of the maps $\rho\mapsto \tr_{\C^2}(P_i \rho P_i)$ and 
    $\rho\mapsto \calD_{\leq t}(\tr_{\C^2}(P_0 \rho P_0))$  is also CP as a composition of CP maps (since $\calD_{\leq t}$ is a quantum channel, it is CPTP). 
    This shows that $\Tilde{\calD}$ is CP. Furthermore, since $\calD$ is CPTP we have
       \[
    \tr(\Tilde D(\rho)) = \tr\p{\calD_{\leq t}(\tr_{\C^2}(P_0\rho P_0))}+\tr\p{\tr_{\C^2}(P_1\rho P_1)}=\tr\p{P_0\rho P_0}+\tr\p{P_1\rho P_1}=\tr(\rho),
      \]
    To conclude, we have shown that $\Tilde{\calD}$ is CPTP. 
        We now define the decoding operation to be $\calD=\Tilde{\calD}\circ \calM$ and note that it is CPTP as a composition of two CPTP operators.

{\em Action on noisy states}.        
    For any state $\rho\in D(\calH_{q,N})$ and error operator $A_{\ur}\in \calS_{q,r}$, 
    \[A_{\ur}  \rho A_{\ur}^\dag \in \begin{cases}
        \calH_{q,N-r} & r\leq N \\
        \set{0} & r>N, 
    \end{cases}\]
    and in particular, 
    \[P_{\leq t} (A_{\ur}  \rho A_{\ur}^\dag) P_{\leq t} =\begin{cases}
        A_{\ur}  \rho A_{\ur}^\dag & r\leq t\\
        0 & r>t.
    \end{cases}\]
    Thus, 
    \begin{align*}
        \calM \circ \calN (\rho)&=\sum_{\substack{\ur\in \calS_{q,r}\\r\leq t}}  \calM(A_{\ur}  \rho A_{\ur}^\dag) + \sum_{\substack{\ur\in \calS_{q,r}\\r> t}}  \calM(A_{\ur}  \rho A_{\ur}^\dag)\\
        &=\sum_{\substack{\ur\in \calS_{q,r}\\r\leq t}}  A_{\ur}  \rho A_{\ur}^\dag\otimes \ket{0}\bra{0} + \tr\biggl(\sum_{\substack{\ur\in \calS_{q,r}\\r> t}}  A_{\ur}  \rho A_{\ur}^\dag\biggr)\cdot \tau \otimes\ket{1}\bra{1}
        \\&= p_{N,t}\cdot \calN_{\leq t}(\rho)\otimes \ket{0}\bra{0} + (1-p_{N,t})\cdot \tau\otimes \ket{1}\bra{1},
    \end{align*}
    where the last equality follows \eqref{eq:Trunc<t} and the completeness condition $\sum_{\ur\in\Z_0^q} A_{\ur}^\dag A_{\ur}=I$.
We can now derive the action of the decoder on noisy states on the output of $\calN$: 
    \begin{align}
       \notag \calD\circ \calN (\rho)&=\Tilde{D}\p{\calM \circ \calN (\rho)}\\
        &\notag=\Tilde{D}\p{ p_{N,t}\cdot \calN_{\leq t}(\rho)\otimes \ket{0}\bra{0} + (1-p_{N,t})\cdot \tau\otimes \ket{1}\bra{1}}\\
        &= p_{N,t}\cdot  \calD_{\leq t}\circ \calN_{\leq t}(\rho)+(1-p_{N,t})\cdot \tau.\label{eq:actionDN}
    \end{align}
    It is easy to verify that the unique linear map that acts on density operators as \eqref{eq:actionDN} acts on general linear maps $X\in L(Q)$ as 
       \[
    \calD\circ \calN (X)= p_{N,t}\cdot  \calD_{\leq t}\circ \calN_{\leq t}(X)+(1-p_{N,t})\tr(X)\cdot \tau.
      \]
 where $\calN(X)$ denotes the action of the AD channel on a general linear map.   
{\em Computing the entanglement fidelity.} For a density operator $\rho\in D(Q)$, let $\psi$ be its purification. Let $\ket{\psi}=\sum_{i=0}^{K-1} p_i \ket{i}\otimes \ket{i'}$ be a Schmidt decomposition of $\ket{\psi}$, where $\{\ket i\}$ is an orthonormal basis of $Q$,  $p_i\geq 0$ for all $i$,
and $\sum_i p_i^2=1$. We set out to compute the entanglement fidelity $\calF_e$, \eqref{eq: EF}. First, observe that
    \begin{align*}
        \calF&\p{I_{L(Q)}\otimes \calD\circ \calN(\ket{\psi}\bra{\psi}),\ket{\psi}\bra{\psi}}\\
        &\qquad =\bra{\psi} I_{L(Q)}\otimes\calD\circ \calN(\ket{\psi}\bra{\psi}) \ket{\psi}\\
        &\qquad =\bra{\psi} I_{L(Q)}\otimes\calD\circ \calN\Big(\sum_{i,j}p_ip_j \ket{i}\bra{j}\otimes \ket{i'}\bra{j'}\Big) \ket{\psi}\\
        &\qquad =\bra{\psi}\Big(\sum_{i,j}p_ip_j \ket{i}\bra{j}\otimes  \Big(p_{N,t\cdot }\calD_{\leq t}\circ \calN_{\leq t}\p{\ket{i'}\bra{j'}}+(1-p_{N,t})\cdot\tau\Big)\Big)\ket{\psi}\\
        &\qquad =\bra{\psi}\Big(\sum_{i,j}p_ip_j \ket{i}\bra{j}\otimes  \Big(p_{N,t\cdot }\calD_{\leq t}\circ \calN_{\leq t}\p{\ket{i'}\bra{j'}}+(1-p_{N,t})\tr(\ket{i'}\bra{j'})\cdot\tau\Big)\Big) \ket{\psi}\\
        &\qquad =  p_{N,t}\cdot \bra{\psi} I_{L(Q)}\otimes \calD_{\leq t}\circ \calN_{\leq t}(\ket{\psi}\bra{\psi}) \ket{\psi}+ (1-p_{N,t})\cdot\bra{\psi} \Tilde{\tau} \ket{\psi},\\
\intertext{where}
     &\hspace*{1.7in}\Tilde{\tau}=\rho_{\psi}\otimes \tau, \quad  \rho_{\psi}=\sum_{i=0}^{K-1}p_i^2\cdot \ket{i}\bra{i}. 
       \end{align*}
Since $\Tilde{\tau}$ is positive semidefinite, we have $\bra{\psi}\Tilde{\tau}\ket{\psi}\ge 0$, so the second term can be discarded, and we obtain
    \begin{align*}
        \calF\p{I_{L(Q)}\otimes \calD\circ \calN(\ket{\psi}\bra{\psi}),\ket{\psi}\bra{\psi}}&\geq p_{N,t}\cdot \bra{\psi} I_{L(Q)}\otimes \calD_{\leq t}\circ \calN_{\leq t}(\ket{\psi}\bra{\psi}) \ket{\psi}\\
        &\geq p_{N,t}\cdot (1-\varepsilon^2),
    \end{align*}
where the last inequality follows from the assumption of the theorem (that the entanglement fidelity for $\calD_{\leq t}$ is close to 1 under the action
of the truncated noise $\calN_{\leq t}$. Since this is true for an arbitrary pure state in $C^{\otimes 2}$, it is also true for the infimum in \eqref{eq: EF}, completing the forward part.

    Let us now prove the converse. Assume that $Q$ is $\varepsilon$-AQECC for $\calN$, with a decoding operation $\calD$ such that $\calF_{e}(\calD\circ \calN|_Q,I_{L(Q)} )\geq 1-\varepsilon^2$. We will show that 
    \[\calF_{e}(\calD\circ \calN_{\leq t}|_Q,I_{L(Q)} )\geq 1-\frac{\varepsilon^2}{p_{N,t}}.\]
We will also use the ``complementary channel'', with more than $t$ errors, defined as
    \begin{align}
        \calN_{>t}(\rho)=\frac{1}{1-p_{N,t}}\cdot \sum_{\substack{\ur \in \calS_{q,r}\\ r>t }}A_{\ur}\rho A_{\ur}^\dag.\label{eq:Trunc>t}
    \end{align}
By the same argument as in Proposition~\ref{obs:quantumChannelTrunc}, we claim that $\calN_{>t}$ is a quantum channel on $\calH_{q,N}$ (as for any $\rho\in L(Q)$ we have $\tr(\rho)=1-p_{N,t}$). We also observe that  
       \[
       \calN(\rho)=p_{N,t}\cdot \calN_{\leq t}(\rho)+(1-p_{N,t})\cdot \calN_{>t}(\rho).
       \]
Fix a pure state $\ket{\psi}\in Q^{\otimes 2}$. We will show that
    \begin{align}
        \calF\p{I_{L(Q)}\otimes \calD\circ \calN_{\leq t}(\ket{\psi}\bra{\psi}),\ket{\psi}\bra{\psi}}\geq 1-\frac{\varepsilon^2}{p_{N,t}}.\label{eq:fiDelConverse}
    \end{align}
Note that we may write $I_{L(Q)}\otimes \calD\circ \calN=(I_{L(Q)}\otimes \calD)\circ (I_{L(Q)}\otimes \calN) $. Thus, since fidelity is linear whenever one of the involved states is pure, we have: 
    \begin{align*}
        1-\varepsilon^2&\leq \calF\p{I_{L(Q)}\otimes \calD\circ \calN(\ket{\psi}\bra{\psi}),\ket{\psi}\bra{\psi}}\\&
        =\calF\p{(I_{L(Q)}\otimes \calD)\circ (I_{L(Q)}\otimes \calN)(\ket{\psi}\bra{\psi}),\ket{\psi}\bra{\psi}}\\
        &=\calF\p{(I_{L(Q)}\otimes \calD)\circ (I_{L(Q)}\otimes (p_{N,t}\cdot \calN_{\leq t}+(1-p_{N,t})\calN_{>t}))(\ket{\psi}\bra{\psi}),\ket{\psi}\bra{\psi}}\\
        &=p_{N,t}\cdot \calF\p{(I_{L(Q)}\otimes \calD)\circ (I_{L(Q)}\otimes  \calN_{\leq t})(\ket{\psi}\bra{\psi}),\ket{\psi}\bra{\psi}}\\
        &\quad +(1-p_{N,t})\cdot \calF\p{(I_{L(Q)}\otimes \calD)\circ (I_{L(Q)}\otimes  \calN_{> t})(\ket{\psi}\bra{\psi}),\ket{\psi}\bra{\psi}}\\
        &\leq p_{N,t}\cdot \calF\p{I_{L(Q)}\otimes \calD\circ\calN_{\leq t} (\ket{\psi}\bra{\psi}),\ket{\psi}\bra{\psi}}+(1-p_{N,t}),
    \end{align*}
    where the last inequality follows since $I_{L(Q)}\otimes \calD \circ \calN_{>t}$ is a quantum channel, and in particular  $I_{L(Q)}\otimes \calD\circ \calN_{\leq t} (\ket{\psi}\bra{\psi})$ is a quantum state (and therefore the fidelity is bounded by $1$). Rearranging the above inequality, we obtain \eqref{eq:fiDelConverse}, which completes the proof.

\subsection{Typical vectors in the simplex}

For $\ux\in \calS_{q,N}$ let 
     \[
     \supp(\ux)=\ppp{i\in [q] ~:~ x_i\neq 0}, \quad \norm{\ux}_\infty=\max_{i\in [q]} |x_i|. 
     \]

\begin{lemma}[Tail bound for $\infty$-norm on the simplex]\label{lem:Tailmax}    Let $\su{X}=(\s{X}_0,\dots,\s{X}_{q-1})$ be a uniform random vector in $\calS_{q,N}$, $q\in \N$, and let $B\in \{0,1,\dots,N\}$ be an integer. Then 
    \[\P\pp{\norm{\su{X}}_{\infty} \geq B}\leq q \cdot \frac{\binom{N-{B}+q-1}{q-1}}{\binom{N+q-1}{q-1}}. \]
\end{lemma}
\begin{proof}
    Let us fix some $i\in [q]$ and bound the probability $\P\pp{\s{X}_i\geq B}$.   We start by observing that 
    \begin{align*}
        \P\pp{\s{X}_i\geq B}&= \frac{\abs{\set{\un\in \calS_{q,N} ~:~ n_i\geq B}}}{\abs{\calS_{q,N}}}=\frac{\binom{N-{B}+q-1}{q-1}}{\binom{N+q-1}{q-1}},
    \end{align*}
    where the last equality  
    follows since once $B$ balls are placed, we are left to distribute
    $N-B$ balls among $q$ bins.
    Next,
    \[
    \P[\max_i \s{X}_i \geq B]\leq \sum_{i=1}^q  \P\pp{\s{X}_i\geq B} = q \cdot \frac{\binom{N-{B}+q-1}{q-1}}{\binom{N+q-1}{q-1}}. \qedhere
    \]
\end{proof}

\begin{lemma}[Tail bound in the support size  on the simplex]\label{lem:inftyNormBound}
       Let $\su{X}=(\s{X}_0,\dots,\s{X}_{q-1})$ be a uniform random vector in $\calS_{q,N}$, $q\in \N$, and let $\varepsilon>0$. Let $\s{S}=\frac{1}{q}|\supp(\su{X})|$ Then 
    \begin{equation}\label{eq: support}
    \P\pp{\abs{\s{S}- \frac{N}{N+q-1}}>\varepsilon}=\P\pp{\abs{\s{S}-\E[S]}>\varepsilon}\leq \frac{N(N-1)(q-1)}{\varepsilon^2 q(N+q-1)^2(N+q-2)}. 
    \end{equation}
\end{lemma}

\begin{proof}
 Let $\s{Y}=1-\s{S}$, be the normalized number of zero coordinates of $\su{X}$: 
    \[\s{Y}=\frac{1}{q}\sum_{i=0}^{q-1} \s{Y}_i, \quad \s{Y}_i:= \Ind_{\set{\s{X}_i=0}}.\]
    Note that 
    \[\E\pp{\s{Y}_i}=\P\pp{\s{X}_i=0}=\frac{\binom{N+q-2}{q-2}}{\binom{N+q-1}{q-1}}=\frac{q-1}{N+q-1}.\]
    and therefore 
    \[E\pp{\s{S}}=1-\E[\s{Y}]=1-\frac{1}{q}\sum_{i=0}^{q-1}E\pp{\s{Y}_i}=\frac{N}{N+q-1}.\]
    Similarly, we compute the variance:
    \[\var(\s{S})=\var(1-\s{Y})=\var(\s{Y})=\frac{1}{q^2}\Big(\sum_{i=1}^q \var(\s{Y}_i)+\sum_{i\neq j}\cov(\s{Y}_i,\s{Y}_j)\Big).\]
    Since $\s{Y}_i$ is a $\s{Ber}(p)$ random variable with $p=\frac{q-1}{N+q-1}$ we have $\var(\s{Y}_i)=p(1-p)$. For the cross terms, we have:
    \begin{align*}
        \cov(\s{Y}_i,\s{Y}_j)&=\E\pp{\s{Y}_i\s{Y}_j}-E\pp{\s{Y}_i}E\pp{\s{Y}_i}=\P\pp{\s{X}_i=\s{X}_j=0}-\p{\frac{q-1}{N+q-1}}^2\\
        &=\frac{\binom{N+q-3}{q-3}}{\binom{N+q-1}{q-1}}-\p{\frac{q-1}{N+q-1}}^2=-\frac{(q-1)N}{(N+q-1)^2(N+q-2)}
    \end{align*}
after simplifications. This yields
    \begin{align*}
        \var(\s{S})&=\frac{1}{q^2}\p{q\cdot \frac{q-1}{N+q-1}\p{1-\frac{q-1}{N+q-1}}-q(q-1)\cdot\frac{(q-1)N}{(N+q-1)^2(N+q-2)}}\\
        &=\frac{N(q-1)}{q(N+q-1)}\p{\frac{N-1}{(N+q-1)(N+q-2)}}.
    \end{align*}
By Chebyshev's inequality, $\P [\abs{\s{S}- \E[\s{S}]}>\varepsilon]\leq \frac{\var(\s{S})}{\varepsilon^2},$ and substitutions yield \eqref{eq: support}.
\end{proof}

In the next lemma, we compute a bound on the volume of the ball with center of bounded support. 
\begin{lemma}\label{lem:BallSizeBound}
    Let $\un\in \calS_{q,N}$ be a vector with $\abs{\supp(\un)}=m$. Then 
    \[\abs{\sr{B}(\un,r)}\leq \sum_{j=0}^{r}\binom{j+m-1}{m-1} \binom{j+q-1}{q-1}. \]
\end{lemma}

\begin{proof}
    Let $|\supp(\un)|=m$ and let $\un'$ be such that $d(\un,\un')=j$. Let $A_1=\{i:\un_i>\un'_i\}$ and $A_2=\{i:\un_i<\un'_i\}$. Since $\sum_{i}\un_i=\sum_i\un'_i$, we have that $\sum_{i\in A_1}(\un_i-\un'_i)
=\sum_{i\in A_2}(\un'_j-\un_i)=j$. Thus, we must place $j$ units within $A_1\subset\supp(\un)$ and the other $j$ units in $A_2\subset [q]\backslash A_1$. 
Since we do not have control over the size of $A_2$, the best upper bound on the number of points $n'$ that we can get is 
$\binom{j+m-1}{m-1} \binom{j+q-1}{q-1}$. Summing over all $j$, we obtain the claimed estimate.
\end{proof}

\subsection{Combinatorial and probabilistic results}

\begin{lemma}{\rm (The continuous mappting theorem, \cite[Theorem 2.7]{Billingsley1999})}\label{lem:ctsMap}
Let $(\s{S}_n)_n$ be a sequence of random variables with values in a metric space $(M,d)$, and assume that $T:M\to M'$ is a continuous map to another metric space $(M',d')$. Then if $\s{S}_n$ coverage to some variable $\s{S}$ in distribution, then $T(\s{S}_n)$ converge to $T(\s{S})$ in distribution. 
\end{lemma}
Another useful result concerns the convergence of expectations for sequences that converge in distribution: 
\begin{lemma} {\rm \cite[Theorem 3.5]{Billingsley1999}} \label{lem:ConvExp} Let $(\s{S}_n)_n$ be a sequence of random variables that converge in distribution to $\s{S}$. Then $\E\pp{\s{S}_n}\to \E\pp{\s{S}}$ if $(\s{S}_n)_n$ are uniformly integrable, that is, 
\[\sup_n\E\pp{\abs{\s{S}_n}}<\infty.\]
    
\end{lemma}

A function $f:A_1,\dots, A_n $ and $f:A_1\times \cdots \times A_n\to \R$ is said to satisfy the {\em bounded differences property} if changing the value of the $i$th coordinate $x_i$ changes the value of $f$ by at most $c_i$. 
Formally, this means that 
\begin{equation}\label{eq: bounded differences}
\max_{x_i'\in A_i}|f(x_1,\dots,x_{i-1},x_i,x_{i+1},\dots, x_n)-f(x_1,\dots,x_{i-1},x_i',x_{i+1},\dots, x_n)|\le c_i
\end{equation}
for all $i\in \{1,\dots, n\}$.

\begin{lemma}{\rm (McDiarmid's inequality, \cite[Theorem 3.1]{mcdiarmid1989method}; \cite[Sec.~6.1]{Boucheron2013})}\label{lem:McDiarmind}
Let $\s{X}_1,\dots, \s{X}_n$ be independent random variables taking values in $A_1,\dots, A_n $ and $f:A_1\times \cdots \times A_n\to \R$ be a measurable function with bounded differences $c_1,\dots, c_n$. Then 
    \[\P\pp{\abs{f(\su{X})-\E[f(\su{X})]}>\varepsilon}\leq 2\exp\p{-\frac{2\varepsilon^2}{\sum_{i=1}^n c_i^2}}.\]
\end{lemma}

\begin{lemma}{\rm (Hoeffding's inequality, \cite[Theorem 2]{hoeffding1963probability}; \cite[Lemma 2.2]{Boucheron2013})} \label{lem:Hoeffding}
Let $\s{X}_1,\dots, \s{X}_n$ be independent random variables such that $a_i\leq X_i\leq b_i$  with probability $1$. Then 
    \begin{equation}\label{eq: Hoeffding}
    \P\pp{\abs{\sum_{i=1}^n (\s{X_i}-\E[\s{X}_i])}>\varepsilon}\leq 2\exp\p{-\frac{2\varepsilon^2}{\sum_{i=1}^n (b_i-a_i)^2}}.
    \end{equation}
\end{lemma}

Finally, we repeatedly use the asymptotic approximation for binomial coefficients: for a fixed $\alpha$ and $N\to\infty$,
   \begin{align}
    \binom{N}{\floor{\alpha N}} =  2^{N (h_2(\alpha)+o(1))},\label{eq:asymbinom}
\end{align}
where $h_2(\cdot)$ is the binary entropy function; e.g., \cite[Ch.~10, Lemma 7]{macwilliams1977theory}.

\subsection{Additional proofs}\label{sec:additProof}
 \begin{proof}[{Proof of Corollary~\ref{cor:identification}}]
     For $\ur\in\calS_{q,r}$ 
  let $\Pi_{\ur}$ be the orthogonal projection on the space $A_{\ur} (Q_{N})=\{A_{\ur}\ket{\un} ~:~ \un \in C_N\}$, and let $\Pi_{t>1}=I-\sum_{r<t}\Pi_{\ur}$.
     Let  $\ket{\psi}\in Q_N$ be any input state, and let $\ur\in \calS_{q,r}$, $r\leq t$ 
     be fixed. Note that the noisy state $\calN (\ket{\psi}\bra{\psi})$ satisfies
     \begin{align*}
         \Pi_{\ur}\calN (\ket{\psi}\bra{\psi}) \Pi_{\ur}&=\Pi_{\ur}
         \Biggl\{(1-p_{N,t})\cdot \calN_{> t} (\ket{\psi}\bra{\psi})+  \sum_{\substack{\ur' \in \calS_{q,r'}\\
         r'\leq t}}A_{\ur'}\ket{\psi}\bra{\psi}A_{\ur'}^\dag \Biggr\} \Pi_{\ur}
     \end{align*}
    where  the channel $\calN_{>t}$ is defined in  \eqref{eq:Trunc>t}.  We observe that $\calN_{>t}$ maps states in $L(\calH_{q,N})$ to states supported on $\calH_{q,<N-t}\triangleq\sr{Span}\set{\bigcup_{0<i<N-t}\calH_{q,i}}$. We also observe that $A_{\ur} (Q_N)\subseteq \calH_{q,N-r}\subset \calH_{q,<N-t}^{\perp} $ and therefore $\Pi_{\ur}\calN(\ket{\psi}\bra{\psi}) \Pi_{\ur}=0$. By \eqref{eq:condOrth}, any error operators $A_{\ur'}$, $\ur\neq \ur'\in \calS_{q,r}$, $r\leq t$, maps the code to orthogonal subspaces of $\calH_{q}$. In particular, $\Pi_{\ur} A_{\ur'}\ket{\psi}=0$. Combining these observations, we obtain: 
     \begin{align*}
         \Pi_{\ur}\calN (\ket{\psi}\bra{\psi}) \Pi_{\ur}&=(1-p_{N,t})\cdot \overset{0}{\overbrace{\Pi_{\ur} \calN_{> t} (\ket{\psi}\bra{\psi})\Pi_{\ur} }}\\
         &\quad+\sum_{\substack{\ur'\in \calS_{q,r}\\ r'\leq t}} \overset{\delta_{\ur,\ur'}\cdot A_{\ur}\ket{\psi}\bra{\psi}A_{\ur}^\dag}{\overbrace{\Pi_{\ur}A_{\ur'}\calN(\ket{\psi}\bra{\psi})A_{\ur'}^\dag}\Pi_{\ur}}=A_{\ur}\ket{\psi}\bra{\psi}A_{\ur}^\dag.
     \end{align*}
     We also note that the operator $\Pi_{>t}$ is measured with probability:
     \begin{align}\label{eq:InnerPgen}
         \tr\p{\Pi_{>t}\calN (\ket{\psi}\bra{\psi}) \Pi_{>t}}=(1-p_{N,t})\cdot\tr\p{ \calN_{> t} (\ket{\psi}\bra{\psi})}=1-p_{N,t}.
     \end{align}

     We conclude by observing that we do not need the full power of the condition \eqref{eq:condOrth}, only that $\bra{c_i}A_{\ur}A_{\ur'}^\dag \ket{c_j}=0$ holds for $\ur\neq \ur'$. Note that \eqref{eq:InnerPgen} holds in general, without any assumption on $I_i$ and $I_j$. Note that if $\ur\neq  \ur'$ the inner products $\braket{\un-\ur|\un'-\ur'}$  in \eqref{eq:InnerPgen} must be all zero, since $\un \neq \un'+\ur-\ur'$ as $\ur-\ur'$ is a nonzero vector of $\ell_1$ norm at most $2t$ and $d(\un,\un')>t$ by the distance assumption.  
 \end{proof}

 \begin{proof}[{Proof of Proposition~\ref{obs:quantumChannelTrunc}}]
     We note that $\calN_{\leq t}$ is a completely positive map since any map that admits a Kraus representation is such.  To show that it is also trace-preserving, let $X\in L(\calH_{q,N})$ be any operator. It is always possible to decompose it as $X=T_1+i\cdot T_2$, where $T_1$ and $T_2$ are Hermitian. Since $\calH_{q,N}$ is finite-dimensional,
these operators can be written as $T_l=\sum_{j}\lambda_{l,j} \ket{\psi_{l,j}}\bra{\psi_{l,j}} $, $l=1,2$. By Proposition~\ref{prop:ConstExFidel} we have
 \begin{align*}
     \tr\p{\calN_{\leq t}(X)}&=\tr\p{\calN_{\leq t}(T_1+i\cdot T_2)}\\
     &=\frac{1}{p_{N,t}}\cdot \sum_{j}\biggl(\lambda_{1,j}\cdot\sum_{\substack{\ur\in \calS_{q,r}\\r\leq t}}  \bra{\psi_{1,j}}A_{\ur}^\dag A_{\ur} \ket{\psi_{1,j}}+i\cdot \lambda_{2,j}\cdot \sum_{\substack{\ur\in \calS_{q,r}\\r\leq t}}\bra{\psi_{2,j}}A_{\ur}^\dag A_{\ur} \ket{\psi_{2,j}}\biggr)\\
     &=\frac{1}{p_{N,t}}\cdot \sum_{j}\p{\lambda_{1,j}\cdot p_{N,t}+i\cdot \lambda_{2,j}\cdot p_{N,t}}= \tr(T_1)+i\cdot \tr(T_2)=\tr(X). \qedhere
 \end{align*}
 \end{proof}

\end{document}